\documentclass[10pt,draftclsnofoot,onecolumn]{IEEEtran}
\usepackage{amsmath, amsthm, amssymb}
\usepackage{amsfonts}
\usepackage{amssymb}
\usepackage{graphicx}
\usepackage{color}
\usepackage{multirow}
\usepackage{graphicx}
\newcommand{\myhash}{%
  {\settoheight{\dimen0}{C}\kern-.05em\, \resizebox{!}{\dimen0}{\raisebox{\depth}{\#}}}}

\usepackage[numbers,sort&compress]{natbib}

\def\snr{\mathsf{SNR}}
\def\range{\clR}

\def\mindex#1{\index{#1}}

\def\sq{\hbox{\rlap{$\sqcap$}$\sqcup$}}
\def\qed{\ifmmode\sq\else{\unskip\nobreak\hfil
\penalty50\hskip1em\null\nobreak\hfil\sq
\parfillskip=0pt\finalhyphendemerits=0\endgraf}\fi\medskip}

\long\def\defbox#1{\framebox[.9\hsize][c]{\parbox{.85\hsize}{%
\parindent=0pt
\baselineskip=12pt plus .1pt      
\parskip=6pt plus 1.5pt minus 1pt 
 #1}}}

\long\def\beginbox#1\endbox{\subsection*{}%
\hbox{\hspace{.05\hsize}\defbox{\medskip#1\bigskip}}%
\subsection*{}}

\def\endbox{}

\newsavebox{\junk}
\savebox{\junk}[1.6mm]{\hbox{$|\!|\!|$}}

\def\det{{\mathop{\rm det}}}

\def\argmin{\mathop{\rm arg\, min}}

\newcommand{\field}[1]{\mathbb{#1}}

\def\ind{\field{I}}

\def\bP{{\mathbb P}}

\def\bR{{\mathbb R}}

\def\bfA{{\bf A}}
\def\bfB{{\bf B}}
\def\bfC{{\bf C}}
\def\bfD{{\bf D}}

\def\bfG{{\bf G}}

\def\bfI{{\bf I}}

\def\bfS{{\bf S}}

\def\bfd{{\bf d}}

\def\bfg{{\bf g}}

\def\bft{{\bf t}}
\def\bfu{{\bf u}}
\def\bfv{{\bf v}}
\def\bfw{{\bf w}}
\def\bfx{{\bf x}}
\def\bfy{{\bf y}}
\def\bfz{{\bf z}}

\def\frkB{{\mathfrak{B}}}

\def\bfmath#1{{\mathchoice{\mbox{\boldmath$#1$}}%
{\mbox{\boldmath$#1$}}%
{\mbox{\boldmath$\scriptstyle#1$}}%
{\mbox{\boldmath$\scriptscriptstyle#1$}}}}

\def\bfmY{\bfmath{Y}}

\def\bfmhhaY{\bfmath{\hhaY}} 
\def\bfmhhaY{\hbox to 0pt{$\widehat{\bfmY}$\hss}\widehat{\phantom{\raise 1.25pt\hbox{$\bfmY$}}}}

\def\til={{\widetilde =}}

\def\clB{{\cal B}}

\def\clI{{\cal I}}

\def\clN{{\cal N}}
\def\clO{{\cal O}}

\def\clR{{\cal R}}
\def\clS{{\cal S}}

\def\clV{{\cal V}}
\def\clW{{\cal W}}
\def\clX{{\cal X}}

\def\clZ{{\cal Z}}

 \def\FRAC#1#2#3{\genfrac{}{}{}{#1}{#2}{#3}}

\def\ddtp{{\mathchoice{\FRAC{1}{d^{\hbox to 2pt{\rm\tiny +\hss}}}{dt}}%
{\FRAC{1}{d^{\hbox to 2pt{\rm\tiny +\hss}}}{dt}}%
{\FRAC{3}{d^{\hbox to 2pt{\rm\tiny +\hss}}}{dt}}%
{\FRAC{3}{d^{\hbox to 2pt{\rm\tiny +\hss}}}{dt}}}}

\def\half{{\mathchoice{\FRAC{1}{1}{2}}%
{\FRAC{1}{1}{2}}%
{\FRAC{3}{1}{2}}%
{\FRAC{3}{1}{2}}}}

\def\average#1,#2,{{1\over #2} \sum_{#1}^{#2}}

\def\eye(#1){{\bf(#1)}\quad}

\newtheorem{theorem}{{\bf Theorem}}[section]

\newtheorem{proposition}[theorem]{{\bf Proposition}}
\newtheorem{lemma}[theorem]{{\bf Lemma}}

\def\eq#1/{(\ref{e:#1})}

\newcommand{\inp}[2]{{\langle #1, #2 \rangle}}

\newcommand{\beqn}[1]{\notes{#1}%
\begin{eqnarray} \elabel{#1}}

\newcommand{\eeqn}{\end{eqnarray} }

\newcommand{\beq}[1]{\notes{#1}%
\begin{equation}\elabel{#1}}

\newcommand{\eeq}{\end{equation}}

\def\bdes{\begin{description}}
\def\edes{\end{description}}

\newcounter{rmnum}

\newcounter{anum}

{\end{list}}

\def\ass(#1:#2){(#1\ref{#1:#2})}

\def\ritem#1{
\item[{\sf \ass(\current_model:#1)}]
}

\newenvironment{recall-ass}[1]{%
\begin{description}
\def\current_model{#1}}{
\end{description}
}

\long\def\comment#1{}

\newfont{\bbb}{msbm10 scaled 700}

\newfont{\bb}{msbm10 scaled 1100}

\renewcommand{\det}{{\hbox{det}}}

\newcommand{\transp}{{\sf T}}

\newlength{\noteWidth}
\long\def\notes#1{\ifinner
             {\tiny #1}
             \else
              \marginpar{\parbox[t]{\noteWidth}{\raggedright\tiny #1}}
               \fi}

\usepackage{tikz}
\usepackage{tikz}
\usepackage{pgfplots}
\pgfplotsset{compat=newest}
\usetikzlibrary{patterns}
\usetikzlibrary{positioning}
\usetikzlibrary{datavisualization}
\usetikzlibrary{datavisualization.formats.functions}
\usetikzlibrary{backgrounds}
\usetikzlibrary{shapes,snakes}

\newcounter{ass}
\newenvironment{assumption}{\begin{list}{{\upshape (P\arabic{ass}) \ }}{\usecounter{ass}
\setlength{\leftmargin}{14pt}
\setlength{\rightmargin}{12pt}
\setlength{\itemindent}{-1pt}
}}{\end{list}}
\def\Ass#1{{(P\ref{ass:#1})}}

\begin{document}

\title{Unlabeled Sensing with Random Linear Measurements}

\author{\IEEEauthorblockN{Jayakrishnan Unnikrishnan,
Saeid Haghighatshoar and Martin Vetterli\footnote{J.U. (dr.j.unnikrishnan@ieee.org) was with the Audiovisual Communications Lab, Ecole Polytechnique F\'{e}d\'{e}rale de Lausanne, Lausanne, Switzerland, and is now with the Software Science \& Analytics Group,
General Electric Global Research, Niskayuna, NY, USA. S.H. (saeid.haghighatshoar@tu-berlin.de) was with the Information Processing Group, Ecole Polytechnique F\'{e}d\'{e}rale de Lausanne, Lausanne, Switzerland, and is now with Communications and Information Theory Group, Technische Universit\"{a}t Berlin, Berlin, Germany. M.V. (martin.vetterli@epfl.ch) is with the Audiovisual Communications Lab, Ecole Polytechnique F\'{e}d\'{e}rale de Lausanne, Lausanne, Switzerland.\\ Portions of the results presented here were presented in an abridged form in \cite{unnhagvet15}.\\ This research was supported by Swiss National Science Foundation under grant number 200021-146423 and by ERC Advanced Investigators Grant 247006 on Sparse Sampling Theory, Algorithms and Applications (SPARSAM).
}}}


\maketitle

\begin{abstract}
We study the problem of solving a linear sensing system when the observations are unlabeled.
Specifically we seek a solution to a linear system of equations $\bfy =\bfA\bfx$ when the order of the observations in the vector $\bfy$ is unknown.
Focusing on the setting in which $\bfA$ is a random matrix with i.i.d. entries, we show that if the sensing matrix $\bfA$ admits an oversampling ratio of $2$ or higher, then with probability $1$ it is possible to recover $\bfx$ exactly without the knowledge of the order of the observations in $\bfy$.
Furthermore, if $\bfx$ is of dimension $K$, then any $2K$ entries of $\bfy$ are sufficient to recover $\bfx$.
This result implies the existence of deterministic unlabeled sensing matrices with an oversampling factor of $2$ that admit perfect reconstruction.
The result is universal in that recovery is guaranteed for all possible choices of $\bfx$. 
While the proof is constructive, it uses a combinatorial algorithm which is not practical, leaving the question of complexity open. 
We also analyze a noisy version of the problem and show that local stability is guaranteed by the solution.
{In particular, for every $\bfx$, the recovery error tends to zero as the signal-to-noise-ratio tends to infinity.}
The question of universal stability is unclear.
We also obtain a converse of the result in the noiseless case: If the number of observations in $\bfy$ is less than $2K$, then with probability $1$, universal recovery fails, i.e., with probability $1$, there exists distinct choices of $\bfx$ which lead to the same unordered list of observations in $\bfy$. 
In terms of applications, the unlabeled sensing problem is related to data association problems encountered in different domains including robotics where it is appears in a method called ``simultaneous localization and mapping" (SLAM), multi-target tracking applications, and in sampling signals in the presence of jitter.

\end{abstract}


\section{Introduction}\label{sec:intro}


Linear sensing and monitoring systems in several disciplines \cite{pravet08,robtre80,aka12,vetkovgoy14} rely on solving a linear system of equations of the form
$
\bfy = \bfA \bfx
$
where $\bfy \in \bR^N$ is an observation vector, $\bfA$ is an $N \times K$ measurement matrix and $\bfx \in \bR^K$ is an unknown system state. 
If these equations represent the true relationship between $\bfy$ and $\bfx$, we know from the basic results of linear algebra that $\bfx$ can be retrieved exactly provided $\bfA$ has rank equal to $K$.
The minimum value for $N$ under which this condition is  satisfied is $N = K$.
We know also that, in the absence of further information on $\bfx$, it is impossible to recover $\bfx$ from $\bfy$ without $\bfA$ having full rank $K$.

In some linear sensing systems it may be practically difficult or impossible to register the entries of $\bfy$ in the correct order. 
In the extreme scenario, one might have access to all the entries of $\bfy$ but not their labels, i.e., one might not know which values correspond to which locations within the vector $\bfy$.
Equivalently, one only has access to $\bfy= {\bf \Pi} \bfA \bfx $ where ${\bf \Pi}$ is an unknown permutation matrix.
In this paper we focus on such unlabeled sensing systems and discuss conditions under which $\bfx$ can be recovered from the unlabeled entries of $\bfy$.

An illustration of unlabeled sensing and a comparison with compressed sensing \cite{don06,cantao06} is provided in Fig. \ref{fig:csvus}.
In the compressed sensing framework, there is no unknown  permutation ${\bf \Pi}$  and thus no ambiguity in the order of the entries in $\bfy$.
The challenge in these problems is to identify the positions and entries of the non-zero entries of $\bfx$, assuming that $\bfx$ has a fixed known number of non-zero entries.
In unlabeled sensing, there is no assumption on the sparsity of $\bfx$, but the order of the entries in $\bfy$ is unknown.
The challenge is to recover $\bfx$ from the unlabeled entries of $\bfy$. 

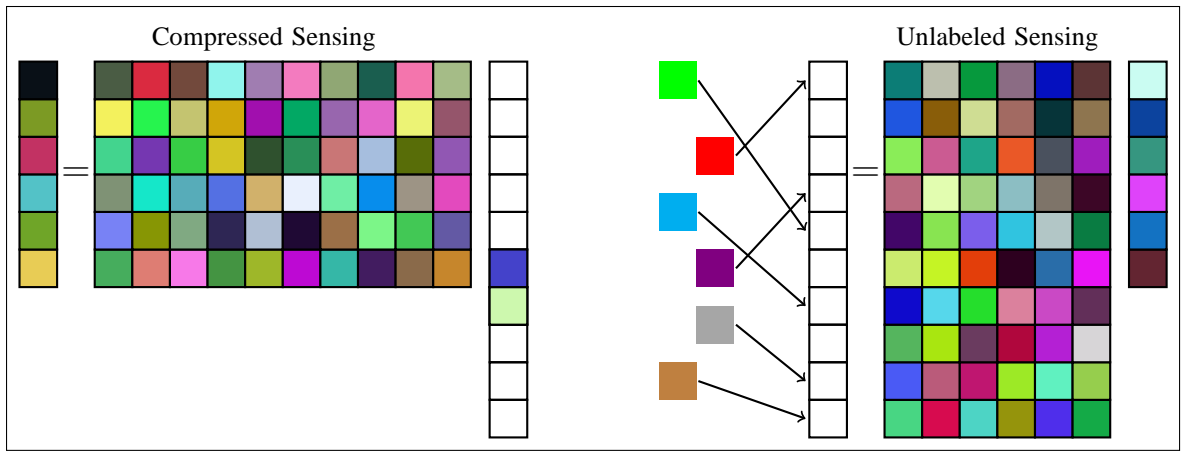
\begin{figure}[h]
\centering
\begin{tikzpicture}[scale=2, show background rectangle]

\begin{scope}[scale=0.25]

\foreach \y in {-3, -2, ..., 2}{
	 \edef\R{\pdfuniformdeviate 255}
    	 \edef\G{\pdfuniformdeviate 255}
    	 \edef\B{\pdfuniformdeviate 255}
     	\xdefinecolor{MyColor}{RGB}{\R,\G,\B}
	\filldraw[draw=black, thick, fill=MyColor] (0,\y) rectangle (1,\y+1);
}

\draw (1.5,0) node {{\Large $=$}};

\foreach \x in {2,3,...,11}{
	\foreach \y in {-3,-2,..., 2}{
	 	\edef\R{\pdfuniformdeviate 255}
    		 \edef\G{\pdfuniformdeviate 255}
    		 \edef\B{\pdfuniformdeviate 255}
     		\xdefinecolor{MyColor}{RGB}{\R,\G,\B}
		\filldraw[draw=black, thick, fill=MyColor] (\x,\y) rectangle (\x+1,\y+1);
	}
}
		
\foreach \y in {2,1,...,-7}{
	\draw[thick] (12.5,\y) rectangle (13.5,\y+1);
}

\foreach \num in {1,2,3}{
	\edef\y{\pdfuniformdeviate 9}
	
	\edef\R{\pdfuniformdeviate 255}
    	\edef\G{\pdfuniformdeviate 255}
    	\edef\B{\pdfuniformdeviate 255}
     	\xdefinecolor{MyColor}{RGB}{\R,\G,\B}
	
	\filldraw[draw=black, thick, fill=MyColor] (12.5,\y-8) rectangle (13.5,\y-7);
}

\draw (6.5,3.6) node {Compressed Sensing};

\end{scope}


\begin{scope}[xshift=5cm, scale=0.25]

\foreach \y in {-3, -2, ..., 2}{
	 \edef\R{\pdfuniformdeviate 255}
    	 \edef\G{\pdfuniformdeviate 255}
    	 \edef\B{\pdfuniformdeviate 255}
     	\xdefinecolor{MyColor}{RGB}{\R,\G,\B}
	\filldraw[draw=black, thick, fill=MyColor] (9.5,\y) rectangle (10.5,\y+1);
}

\draw (2.5,0) node {{\Large $=$}};

\foreach \y in {-7,-6,...,2}{
	\foreach \x in {3,4,...,8}{
	 	\edef\R{\pdfuniformdeviate 255}
    		 \edef\G{\pdfuniformdeviate 255}
    		 \edef\B{\pdfuniformdeviate 255}
     		\xdefinecolor{MyColor}{RGB}{\R,\G,\B}
		\filldraw[draw=black, thick, fill=MyColor] (\x,\y) rectangle (\x+1,\y+1);
	}
}
		
\foreach \y in {2,1,...,-7}{
	\draw[thick] (1,\y) rectangle (2,\y+1);
}

\draw (6,3.6) node {Unlabeled  Sensing};

\edef\X{-2}
\fill[fill=red] (\X,0) rectangle (\X+1,1);
\draw[->, thick] (\X+1.05,0.5) -- (0.9,2.5);

\fill[fill=green] (\X-1,2) rectangle (\X,3);
\draw[->, thick] (\X+0.05,2.5) -- (0.9,-1.5);

\fill[fill=cyan] (\X-1,-1.5) rectangle (\X,-0.5);
\draw[->, thick] (\X+0.05,-1) -- (0.9,-3.5);

\fill[fill=violet] (\X,-3) rectangle (\X+1,-2);
\draw[->, thick] (\X+1.05,-2.5) -- (0.9,-0.5);

\fill[fill=gray!70] (\X,-4.5) rectangle (\X+1,-3.5);
\draw[->, thick] (\X+1.05,-4) -- (0.9,-5.5);

\fill[fill=brown] (\X-1,-6) rectangle (\X,-5);
\draw[->, thick] (\X+0.05,-5.5) -- (0.9,-6.5);

\end{scope}
\end{tikzpicture} 
\caption{Comparison of compressed sensing and unlabeled sensing.}
\label{fig:csvus}
\end{figure}

Unlabeled sensing has potential applications in a number of different fields.
Consider the following example.
You are blindfolded in a room, and the floor is not flat but a 3 dimensional terrain model. 
You can ÔÕsampleÕÕ the height, but you donÕt know where you take the samples.
Is it possible, under some assumption about the terrain model, to recover the location of the samples and the shape of the terrain? 
This is related to a celebrated problem in robotics called ÔÕsimultaneous location and mappingÕÕ (SLAM) \cite{thrleo08}. 
Similar data-association problems also arise in the task of assigning observations to targets in multi-target tracking problems that arise in radar applications \cite{bla86}.
More generally, consider the problem of reconstructing a spatial field from samples.
Let $\bfx$ denote the representation of the field in some $K$-dimensional basis.
Each measurement can be interpreted as an inner product of $\bfx$ with a ``sampling vector'' unique to the location where the sample was taken.
Consider a mobile sensing scheme \cite{unnvet13,unnvet13b} where a moving sensor samples the field at $N$ different locations. 
Further suppose that the mobile sensor does not have access to accurate spatial measurements, although the set of $M$ potential sampling locations and the sampling vectors corresponding to the potential locations are known a priori.
The field reconstruction problem one faces in this situation is precisely the unlabeled sensing problem studied in this paper.
A similar situation arises in time-domain sampling in the presence of clock jitter \cite{bal62} which makes it impossible to associate sampled observations to the correct time indices. 
There is some prior work on reconstruction of bandlimited signals from samples at unknown locations.
In \cite{bro07} an approximate solution to this problem is proposed under the setting of continuous-time measurements and bandlimited signals.
In \cite{marvet00}, an iterative procedure to reconstruct discrete-time bandlimited signals is proposed.
Our work differs from that of these papers in that we do not restrict ourselves to a bandlimited signal model.
Our main results are focused on the setting in which the sampling vectors are randomly distributed.
In such settings we show that an exact solution to the unlabeled sensing problem is possible when we take twice as many samples as required in classic labeled sensing. 

The basic unlabeled sensing problem can be mathematically stated as follows.
Suppose
\begin{equation}
\bfy =   \bfB \bfx \mbox{ with } \bfB = {\bf \Pi} \bfA, \label{eqn:permobs}
\end{equation}
where $\bfA$ is a known matrix and ${\bf \Pi}$ is an unknown permutation matrix.
The goal is to recover $\bfx$ given the observation vector $\bfy$, or equivalently, to recover $\bfx$ given the unlabeled entries of the vector $\bfA \bfx$.

As a simple illustration of unlabeled sensing consider the case of $K=2$ and $N=3$ with 
$
\bfA = \left[ \begin{array}{ccc}
1 &0  & 0\\
0 & 1 &1\end{array} \right]^\transp.
$ 
In this case, the entries of $\bfy$ are $\bfx_1$, $\bfx_2$, $\bfx_2$, and thus $\bfx_2$ can be identified as the entry of $\bfy$ that occurs twice and $\bfx_1$ as the remaining entry of $\bfy$.\footnote{For $K > 2$, an analogous construction of $\bfA$ can be constructed by combining $i$ repetitions of the $i$-th row of $I_K$, the $K \times K$ identity matrix, for all $1 \leq i \leq N$. 
In this case $N = K(K+1)/2$.
This again leads to a system of equations such that $\bfx$ can be recovered from the unlabeled entries of $\bfy$.
In this case the labels of the observations are recovered from the number of repetitions of the observations.} 
This example can be further generalized by replacing the third row with an unbalanced convex combination of the first two rows.
Consider the matrix 
\begin{eqnarray} \label{eqn:Bconvexample}
\bfA = \left[ \begin{array}{cc}
1 &0  \\
0 & 1 \\
\alpha & (1-\alpha) \end{array} \right].
\end{eqnarray}
where $\alpha \in [0, 1] \setminus \{0.5\}$.
In this case the entries of $\bfy$ are $\bfx_1$, $\bfx_2$, $\alpha \bfx_1 + (1-\alpha) \bfx_2$.
Thus we know immediately that if the entries in $\bfy$ are sorted in increasing order, the middle entry of $\bfy$ is $\alpha \bfx_1 + (1-\alpha) \bfx_2$, i.e., if $\bfy_{(1)} \leq \bfy_{(2)} \leq \bfy_{(3)}$ represent the ordered entries of $\bfy$, then
$
\bfy_{(2)} = \alpha \bfx_1 + (1-\alpha) \bfx_2.
$
Moreover, if $\alpha \in (0.5,1]$, then $\bfx_1$ is the remaining entry of $\bfy$ that is closer to $\bfy_{(2)}$ and $\bfx_2$ is the remaining entry of $\bfy$, i.e.,
$
\bfx_1 = \bfy_{(i)} \mbox{ and } \bfx_2 = \bfy_{(j)} 
$
where  $i,j \mbox{ satisfy } |\bfy_{(i)} - \bfy_{(2)}| \leq |\bfy_{(j)} - \bfy_{(2)}|, i,j \in \{1,3\}.$
Similarly if $\alpha \in [0, 0.5) $, then 
$
\bfx_1 = \bfy_{(j)} \mbox{ and } \bfx_2 = \bfy_{(i)} 
$
$\mbox{where } i,j \mbox{ satisfy } |\bfy_{(i)} - \bfy_{(2)}| \leq |\bfy_{(j)} - \bfy_{(2)}|, i,j \in \{1,3\}.$
Thus whenever $\bfA$ has the form of \eqref{eqn:Bconvexample}, the vector $\bfx$ can be uniquely recovered from the unlabeled entries of $\bfy$.
However, in general this property ceases to hold if the last row of $\bfA$ is replaced with a random vector in $\bR^2$.
As an example suppose 
$
\bfA = \left[ \begin{array}{ccc}
1 &0 &1 \\
0 & 1 &2\end{array} \right]^\transp.
$
In this case, it can be verified that
\[
\bfA  \left[ \begin{array}{c}
1  \\
-3 \end{array} \right] =  \left[ \begin{array}{c}
1   \\
-3  \\
-5 \end{array} \right] \mbox{ and } \bfA  \left[ \begin{array}{c}
-5 \\
1 \end{array}\right] =  \left[ \begin{array}{c}
-5   \\
1  \\
-3 \end{array} \right].
\]
Thus the choice of vectors  
$\left[ \begin{array}{cc}
1 & -3 \end{array} \right]^\transp$ 
and 
$ \left[ \begin{array}{cc}
-5& 1 \end{array}\right]^\transp$ 
for $\bfx$ lead to the same entries $\{-5,-3,1\}$ in $\bfy$. 
Thus if the true $\bfx$ were one of these vectors, it would be impossible to determine the value of $\bfx$ given the unlabeled entries of $\bfy$.
However, now suppose we append another randomly chosen row to $\bfA$.
For example, if 
$
\bfA = \left[ \begin{array}{cccc}
1 &0  & 1 &1\\
0 & 1 &2&-1\end{array} \right]^\transp
$
then it can be easily verified that if 
$
\bfA \bfx =  {\bf \Pi} \bfA \bfx'
$
holds for any $\bfx, \bfx' \in \bR^4$ and any permutation matrix ${\bf \Pi}$, then $\bfx = \bfx'$.
Thus in this case $\bfx$ is uniquely determined by the unordered entries of $\bfy = {\bf \Pi} \bfA \bfx$.
This property continues to be true w.p. $1$ if the entries in the last two rows of $\bfA$ are replaced with i.i.d. random variables drawn from a continuous probability distribution. 
Furthermore, this property continues to hold w.p. $1$ if all entries of all rows of $\bfA$ are replaced with i.i.d. random variables drawn from a continuous probability distribution. 
Note that in this setting the number of rows in $\bfA$ is equal to $4$, which is exactly twice the number of entries in $\bfx$.
In fact the unique recovery of $\bfx$ from the unlabelled entries of $\bfy$ continues to hold for vectors $\bfx$ of arbitrary length $K$ provided $\bfA$ is a $2K \times K$ random matrix with i.i.d. entries drawn from a continuous probability distribution.

The main result in this paper is on a more general setting in which the $N$ rows of $\bfB$ are randomly drawn without replacement from the rows of a known $M \times K$ matrix $\bfA$ with i.i.d. random entries for some $M \geq N$.
This leads to a variant of \eqref{eqn:permobs}. 
Suppose
\begin{equation}
\bfy = \bfB \bfx \mbox{ where }\bfB=\bfS \bfA \label{eqn:permobs2}
\end{equation}
and  $\bfS$ is a $N \times M$ selection matrix, i.e., the rows of $\bfS$ comprise some $N$ distinct rows of the $M \times M$ identity matrix, arranged in an arbitrary order.
While $\bfA$ is known by assumption, the selection matrix $\bfS$ is unknown. 
The goal is to recover $\bfx$ without the explicit knowledge of $\bfS$.
It is further assumed that the entries of $\bfA$ are random i.i.d. variables drawn from some continuous probability distribution on $\bR$.
Our main result is that for $N \geq 2K$, it is possible to uniquely recover $\bfx$ from $\bfy$ w.p. $1$.
For the particular case of $M=N=2K$ this result implies that, if one is given unlabeled random projections of the vector $\bfx$ with $2K$ random vectors, it is possible to recover $\bfx$ exactly w.p. $1$.
The ratio of $N$ to $K$ can be viewed as an oversampling factor as we essentially use $N$ (unlabeled) samples to solve for $K$ unknowns.
Thus the result for random matrices implies the existence of sampling matrices $\bfB$ with the oversampling factor $\frac{N}{K}=2$ such that perfect recovery of the signal is possible without any knowledge about the order of the samples.
As an illustration of this result consider the case with $K=1$.
In this case, $\bfB$ is a $2 \times 1$ vector with two i.i.d. entries.
Let
$\bfB = \left[\begin{array}{cc}b_1&b_2\end{array}\right]^\transp$. 
With probability $1$ the ratio $\frac{|b_1|}{|b_2|}$ is different from $1$.
Without loss of generality let us assume $|b_1| > |b_2| $.
Then it is clear that the entry in $\bfy$ with the larger magnitude is equal to $b_1 \bfx$ and thus $\bfx$, which is now a scalar, can be recovered by dividing this entry by $b_1$.

In a practical implementation of unlabeled sensing there are other problems that one might need to solve which we do not address here.
For instance, although we showed that recovery of the unknown $\bfx$ is possible from unlabeled measurements, we do not study the problem of designing an efficient algorithm to recover $\bfx$. 
Our solution is to consider all possible permutations of the unlabeled observations which might be prohibitively complex in large dimensional problems. 
We also do not address the question of the computational complexity of the optimal algorithm.
Another practical problem is the requirement of learning the sensing matrix $\bfA$.
An interesting question is whether it is possible to recover $\bfA$ using a training phase with known $\bfx$'s but with unknown selection matrices $\bfS$ that might change arbitrarily from one observation to the next.

Although the problem of unlabeled sensing in the present form is new, similar problems have been studied by various authors in the past.
For example,  in \cite{emidaugri14}, a variant of the problem in \eqref{eqn:permobs} with repeated observations was studied under a sparsity assumption on $\bfx$ in the context of dictionary learning.
The authors assume that the permutation matrix ${\bf \Pi}$ remains invariant for multiple observations so that the effective sensing system is of the same form as  \eqref{eqn:permobs} but with vectors $\bfx$ and $\bfy$ replaced with matrices with $T$ columns each, where $T$ represents the number of observations.
They propose a branch-and-bound algorithm for solving the problem.
The main difference in the framework of that work from ours is the fact that in \cite{emidaugri14} multiple observations are available under the same permutation matrix ${\bf \Pi}$ which, together with the sparsity assumption, simplifies the task of estimating the unknown permutation.
More generally a number of different problems involve the inversion of the effects of an unknown permutation.
A family of such examples is that of de-anonymization attacks \cite{narshm08,naiunnthivet15} on anonymized databases, which is a well-studied problem in privacy applications. 
In special cases computationally tractable solutions to these problems are possible \cite{naiunnthivet15} and in some others relaxation methods are adopted to approximate the solution \cite{fogjenbacdas13}.
Another application that requires inverting the effect of a permutation occurs in communication channels with deletions and transpositions \cite{schzuc99}.
In such applications the decoding task of determining what message was sent based on the received messages, is a finite-alphabet version of the unlabeled sensing problem studied in this paper.



\medskip



\noindent \textit{Contributions}: 
The main contributions of this paper are summarized below. 

\begin{itemize}
\item We introduce the problem of unlabeled sensing and demonstrate that for signals of arbitrary dimensions, there exist sensing matrices that allow perfect recovery of a signal from unlabeled linear measurements. 
Moreover, an oversampling factor of $2$ is sufficient, i.e., there exist sensing matrices with twice as many rows as columns which gurantee perfect recovery from unlabeled measurements.
\item We identify a property of random matrices that implies that random sensing matrices with oversampling factor of $2$ or higher, w.p. $1$, allow universal recovery of signals from unlabeled measurements. Furthermore, we obtain a converse showing that oversampling by at least $2$ is necessary for this property to hold.
\end{itemize}

\noindent \textit{Paper organization}: 
The rest of the paper is organized as follows.
We introduce the formal problem statement and main result in Section \ref{sec:problem}, and then present the recovery algorithm in Section \ref{sec:rec_alg}, 
and the proof of the main result in Section \ref{sec:proof}.
We present a converse result in Section \ref{sec:converse} and analyze local stability in the presence of noise in Section \ref{sec:stability}.
We conclude with some discussion in Section \ref{sec:conc}.

\medskip

\noindent \textit{Notations and terminology}: 
For any positive integer $N$, we use $[N]$ to denote the set $\{1,2,\ldots,N\}$.
We denote matrices using bold capital letters, e.g., $\bfA$, and vectors using bold lower case letters, e.g., $\bfx$.
By default every vector is a column vector, unless explicitly defined to be a row vector.
The usage will be clear from context.
The $i$-th component of any vector $\bfx$, is denoted $\bfx_i$.
For $N$-dimensional row or column vectors $\bfu$ and $\bfv$, the inner product is denoted by $\langle \bfu, \bfv \rangle = \sum_{i=1}^N \bfu_i \bfv_i$.
The notation $\bfv^i$ is used to denote the $i$-th vector in a sequence of vectors.
A similar notation is used for matrices.
With a slight abuse of terminology, for any $P \times K$ matrix $\bfA$, we refer to any matrix obtained by an arbitrary permutation of $Q$ rows of $\bfA$ as a $Q \times K$ submatrix of $\bfA$.
If $\bfA$ and $\bfB$ are $N \times K$ matrices, we denote by $[\bfA,\bfB]$ the $N \times 2K$ matrix obtained by appending the columns of $\bfB$ to the columns of $\bfA$.
If $\bfu$ and $\bfv$ are column vectors, we use $[\bfu; \bfv]$ to denote the column vector obtained by appending the rows of $\bfv$ to the rows of $\bfu$.
For any $M \times N$ matrix $\bfA$ we denote by $\clN(\bfA) := \{\bfx \in \bR^N: \bfA \bfx = 0\}$ the null-space, and by $\clR(\bfA) := \{\bfA \bfx: \bfx \in \bR^N\}$ the range-space of $\bfA$. 
We denote the $N\times N$ identity matrix with $\bfI_N$. {We use calligraphic letters $\clX, \clV, \dots$ for sets, and the abbreviation {``w.p."} for ``with probability".}

\section{Problem Statement and Main Result}\label{sec:problem}
Let $\bfA$ denote an $M \times K$ random matrix whose entries are i.i.d. random variables drawn from a continuous probability distribution\footnote{In other words, the probability density function $f$ is absolutely continuous with respect to the Lebesgue measure on $\bR$.} over $\bR$ with probability density function $f$. 
Assume that $\bfA$ is known.
Let $\bfx \in \bR^K$ be an arbitrary unknown vector and suppose that the inner products of $\bfx$ with $N$ distinct unknown rows of $\bfA$ are known. 
As the rows selected for computing the inner products are not known, the measurements can be viewed as being unlabeled, for we do not know which measurement corresponds to which row of $\bfA$.
In this setting, we consider the problem of recovering $\bfx$ from these $N$ unlabeled measurements.
The main result of this paper is that it is possible to recover $\bfx$ from these unlabeled measurements provided $M \geq N \geq 2K$ as summarized in the following theorem.

\begin{theorem}[Main Result]\label{thm:main_theorem}
Let $\bfA$ denote a known $M \times K$ matrix with i.i.d. random entries drawn from an arbitrary continuous probability distribution $f$ over $\bR$. 
Let $\bfB$ denote an unknown $N \times K$ sub-martix of $\bfA$, i.e., $\bfB$ consists of some $N \leq M$ rows of $\bfA$.
If $N \geq 2K$, then, w.p.  $1$, every $\bfx \in \bR^K$ can be uniquely recovered from the measurements $\bfy = \bfB \bfx$ without the explicit knowledge of $\bfB$. \qed
\end{theorem}

While the result of Theorem \ref{thm:main_theorem} holds for an arbitrary $M \geq N$, let us now consider the result for $M=N \geq 2K$.
In this case the result states that if $\bfA$ is a $N \times K$ random matrix, then w.p. $1$, $\bfx$ can be recovered from the unlabeled entries of 
$\bfA \bfx$.
We refer to this property of $\bfA$ by saying $\bfA$ admits \emph{universal unlabeled recovery.}
The adjective of universal is to emphasize the fact that the recovery  is guaranteed for all $\bfx \in \bR^K$. 
In other words, once $\bfA$ is designed using a random selection, w.p. $1$ every $\bfx$ can be recovered, even for adversarial choices of $\bfx$ that use the knowledge of $\bfA$.

%
%

Another key parameter of interest in the theorem is the oversampling ratio or the ratio of $N$ to $K$.
The result essentially implies that an oversampling factor of $2$ is sufficient for guaranteeing perfect recovery of $\bfx$ from unlabeled random observations.
In other words, random matrices with oversampling factor of $2$ or higher admit universal unlabeled recovery.
In analogy to the random coding argument from information theory \cite{sha48}, this result implies the existence of deterministic matrices $\bfA$ with oversampling factor $2$ and higher that admit universal unlabeled recovery.
However, the nature of our result differs from that of \cite{sha48} in that our result holds for all $K$ and is not an asymptotic one.

The result of Theorem \ref{thm:main_theorem} captures an inherent geometric property of random matrices.
Essentially the result is that when $\bfA$ is an $M \times K$ random matrix, $\bfx$ is uniquely determined given $\bfB \bfx$ where $\bfB = \bfS \bfA$ with $\bfS$ a selection matrix comprising any $N$ distinct rows of the $N \times N$ identity matrix.
Mathematically this can be expressed as:
\[
\bfS^1 \bfA \bfx^1 = \bfS^2 \bfA \bfx^2 \quad \Rightarrow \quad \bfx^1 = \bfx^2
\]
where $ \bfx^1, \bfx^2 \in \bR^K$, and $\bfS^1, \bfS^2$ are any choices of the selection matrices.
In other words, if $ \bfx^1, \bfx^2$ satisfy $\bfS^1 \bfA \bfx^1 = \bfS^2 \bfA \bfx^2$ for any choices of the selection matrices $\bfS^1, \bfS^2$  we must have $\bfx^1 = \bfx^2$.
The conclusion is immediate if $\clR(\bfS^1 \bfA)$ and $\clR(\bfS^2 \bfA)$ intersect only at the origin.
However, this may not be true for all $\bfS^1$ and $\bfS^2$. 
Nevertheless, we show that even if $\clR(\bfS^1 \bfA) \cap \clR(\bfS^2 \bfA)$ is a non-trivial subspace, then the inverse-images of the operators $\bfS^1 \bfA$ and $\bfS^2 \bfA$ evaluated at any point within $\clR(\bfS^1 \bfA) \cap \clR(\bfS^2 \bfA)$ are identical singletons.
Thus, one might not recover $\bfS$ given $\bfy = \bfS \bfA \bfx$, but one can recover $\bfx$.
Another way to state the same result is that 
\begin{equation}
\clN([\bfS^1 \bfA, \bfS^2 \bfA]) \subset \clN([\bfI_K, \bfI_K]) \label{eqn:nullspace}
\end{equation}
where $\bfI_K$ represents the $K \times K$ identity matrix and $\clN$ represents the null-space operator.
%
%

We will now describe the recovery algorithm and then proceed to the proof of the main result.

\section{The Recovery Algorithm}\label{sec:rec_alg}
For simplicity, we ignore the computational complexity and consider a brute force approach to recover the initial signal $\bfx$. 
The recovery algorithm works as follows. 
It sequentially considers every possible arrangement of $N$ rows of $A$ to form candidate matrices $\widehat \bfB$ and tests whether or not there exists a solution $\widehat{\bfx}$ to the equation $\widehat{\bfB} \widehat{\bfx}=\bfy$.
Equivalently it considers all possible candidate selection matrices $\widehat \bfS$ for the true selection matrix given in \eqref{eqn:permobs2}.
The algorithm returns $\widehat{\bfx}$ when it finds a candidate $\widehat \bfB$ that admits a solution to the equation $\widehat{\bfB} \widehat{\bfx}=\bfy$.
Fig.~\ref{samp_fig} shows a block diagram of the unlabeled sampling and the recovery algorithm.
Here $\bfv^i$ is a row vector representing the $i$-th row of the matrix $\bfA$ for $i \in [M]$.


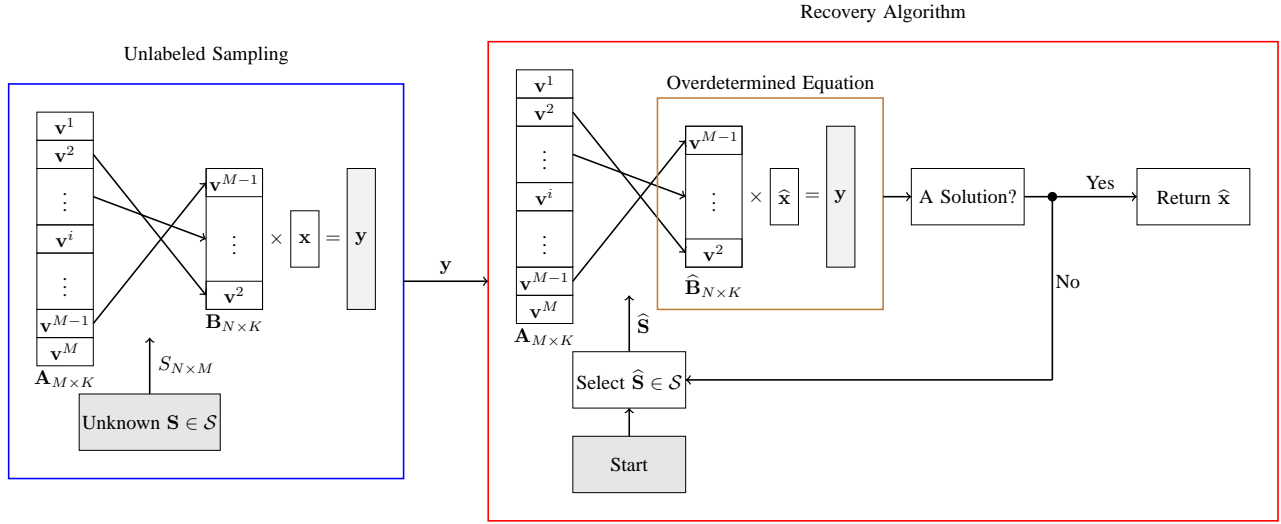
\begin{figure}[h]
\centering
\scalebox{0.75}{
\begin{tikzpicture}

\pgfmathsetmacro{\wid}{0.5}
\pgfmathsetmacro{\len}{1}


\draw [blue, thick] (-\wid, -2*\len) rectangle (6*\len+\wid, 5*\len);
\node at (3*\len, 5*\len+\wid) {Unlabeled Sampling};

\draw [thick, ->] (6*\len+\wid, 1.5*\len) -- (7.5*\len+\wid, 1.5*\len);
\node at (6.75*\len+\wid, 1.5*\len+0.5*\wid) {$\bfy$};

\draw (0,0) rectangle (\len,\wid);
\node at (\len/2,\wid/2) {$\bfv^M$};

\draw (0,\wid) rectangle (\len,2*\wid);
\node at (\len/2,3*\wid/2) {$\bfv^{M-1}$};

\draw (0, 2*\wid) rectangle (\len, 4*\wid);
\node at (\len/2, 3*\wid) {$\vdots$};

\draw (0, 4*\wid) rectangle (\len, 5*\wid);
\node at (\len/2, 4.5*\wid) {$\bfv^i$};

\draw (0, 5*\wid) rectangle (\len, 7*\wid);
\node at (\len/2, 6*\wid) {$\vdots$};

\draw (0, 7*\wid) rectangle (\len, 8*\wid);
\node at (\len/2, 7.5*\wid) {$\bfv^2$};

\draw (0, 8*\wid) rectangle (\len, 9*\wid);
\node at (\len/2, 8.5*\wid) {$\bfv^1$};

\node at (\len/2, -0.5*\wid) {$\bfA_{M\times K}$};


\draw (3*\len, 2*\wid) rectangle (4*\len, 7*\wid);

\draw (3*\len, 2*\wid) rectangle (4*\len, 3*\wid);
\node at (3.5*\len, 2.5*\wid) {$\bfv^{2}$};

\draw (3*\len, 6*\wid) rectangle (4*\len, 7*\wid);
\node at (3.5*\len, 6.5*\wid) {$\bfv^{M-1}$};

\node at (3.5*\len, 4.5*\wid) {$\vdots$};

\draw [thick, ->] (\len,3*\wid/2) -- (3*\len, 6.5*\wid);

\draw [thick, ->] (\len, 7.5*\wid) -- (3*\len, 2.5*\wid);

\draw [thick, ->] (\len, 6*\wid) -- (3*\len, 4.5*\wid);

\filldraw [gray!20, draw=black] (3*\len+\wid/2, -\len-\wid) rectangle (\len-\wid/2, -\len+\wid);
\node at (2*\len, -\len) {Unknown $\bfS \in \clS$};
\draw[->, thick] (2*\len, -\len+\wid) -- (2*\len, -\len+\wid+\len);
\node at (2*\len+1.3*\wid, -\len+2*\wid) {$S_{N\times M}$};

\node at (3.5*\len, 1.5*\wid) {$\bfB_{N \times K}$};

\draw (4.5*\len, 3.5*\wid) rectangle (4.5*\len+\wid, 3.5*\wid+\len);
\node at (4*\len+0.5* \wid, 4.5*\wid) {$\times$};
\node at (4.5*\len+\wid/2, 3.5*\wid+\len/2) {$\bfx$};
\node at (4*\len+2.5* \wid, 4.5*\wid) {$=$};

\filldraw [gray!10, draw=black] (3*\len+\len +3*\wid, 2*\wid) rectangle (4*\len+\len + 3*\wid -\len/2, 7*\wid);
\node at (4.5*\len+5*\wid/2, 3.5*\wid+\len/2) {$\bfy$};

\begin{scope} [xshift=8.5cm, yshift=0.75cm]

\draw[red, thick] (-\wid, -3.5*\len) rectangle  (13.5*\len, 5*\len);
\node at (6.5*\len, 5*\len+\wid) {Recovery Algorithm};


\draw (0,0) rectangle (\len,\wid);
\node at (\len/2,\wid/2) {$\bfv^M$};

\draw (0,\wid) rectangle (\len,2*\wid);
\node at (\len/2,3*\wid/2) {$\bfv^{M-1}$};

\draw (0, 2*\wid) rectangle (\len, 4*\wid);
\node at (\len/2, 3*\wid) {$\vdots$};

\draw (0, 4*\wid) rectangle (\len, 5*\wid);
\node at (\len/2, 4.5*\wid) {$\bfv^i$};

\draw (0, 5*\wid) rectangle (\len, 7*\wid);
\node at (\len/2, 6*\wid) {$\vdots$};

\draw (0, 7*\wid) rectangle (\len, 8*\wid);
\node at (\len/2, 7.5*\wid) {$\bfv^2$};

\draw (0, 8*\wid) rectangle (\len, 9*\wid);
\node at (\len/2, 8.5*\wid) {$\bfv^1$};

\node at (\len/2, -0.5*\wid) {$\bfA_{M\times K}$};


\draw (3*\len, 2*\wid) rectangle (4*\len, 7*\wid);

\draw (3*\len, 2*\wid) rectangle (4*\len, 3*\wid);
\node at (3.5*\len, 2.5*\wid) {$\bfv^{2}$};

\draw (3*\len, 6*\wid) rectangle (4*\len, 7*\wid);
\node at (3.5*\len, 6.5*\wid) {$\bfv^{M-1}$};

\node at (3.5*\len, 4.5*\wid) {$\vdots$};

\draw [thick, ->] (\len,3*\wid/2) -- (3*\len, 6.5*\wid);

\draw [thick, ->] (\len, 7.5*\wid) -- (3*\len, 2.5*\wid);

\draw [thick, ->] (\len, 6*\wid) -- (3*\len, 4.5*\wid);

\node at (3.5*\len, 1.3*\wid) {$\widehat{\bfB}_{N \times K}$};

\draw (4.5*\len, 3.5*\wid) rectangle (4.5*\len+\wid, 3.5*\wid+\len);
\node at (4*\len+0.5* \wid, 4.5*\wid) {$\times$};
\node at (4.5*\len+\wid/2, 3.5*\wid+\len/2) {$\widehat{\bfx}$};

\node at (4*\len+2.5* \wid, 4.5*\wid) {$=$};

\filldraw[gray!10, draw=black] (3*\len+\len +3*\wid, 2*\wid) rectangle (4*\len+\len + 3*\wid -\len/2, 7*\wid);
\node at (4.5*\len+5*\wid/2, 3.5*\wid+\len/2) {$\bfy$};

\draw [color=brown, thick] (2.5*\len, 0.5*\wid) rectangle (6.5*\len, 8*\wid);
\node at (4.5*\len, 8.5*\wid) {Overdetermined Equation};
\draw[->, thick] (6.5*\len, 3.5*\wid+\len/2) -- (7*\len, 3.5*\wid+\len/2);
\draw (7*\len, 2.5*\wid+\len/2) rectangle (9*\len, 4.5*\wid+\len/2);
\node at (8*\len, 3.5*\wid+\len/2) {A Solution?};
\draw[->, thick] (9*\len, 3.5*\wid+\len/2) -- (11*\len, 3.5*\wid+\len/2);
\node at (10.5*\len-0.3*\wid, 4*\wid+\len/2) {Yes};

\draw (11*\len, 2.5*\wid+\len/2) rectangle (13*\len, 4.5*\wid+\len/2);
\node at (12*\len, 3.5*\wid+\len/2) {Return $\widehat{\bfx}$};

\filldraw[thick] (9.5*\len, 3.5*\wid+\len/2) circle (2pt) -- (9.5*\len, -\len);
\node at (9.5*\len+0.55*\wid, 1.5*\wid) {No};

\draw[->, thick] (9.5*\len, -\len) -- (3*\len, -\len);
\draw (3*\len, -\len-\wid) rectangle (\len, -\len+\wid);
\node at (2*\len, -\len) {Select $\widehat{\bfS} \in \clS$};
\draw[->, thick] (2*\len, -\len+\wid) -- (2*\len, -\len+\wid+\len);
\node at (2*\len+0.5*\wid, -\len+2*\wid) {$\widehat{\bfS}$};
\filldraw[gray!20, draw=black] (3*\len, -\len-\wid-\len-\wid) rectangle (\len, -\len+\wid-\len-\wid);
\node at (2*\len, -2*\len-\wid) {Start};
\draw[->, thick] (2*\len, -2*\len) -- (2*\len, -2*\len+\wid);

\end{scope}

\end{tikzpicture}
}
\caption{A block diagram of the unlabeled sampling and recovery algorithm.}
\label{samp_fig}
\end{figure}

Notice that for each candidate $\widehat{\bfB}$ the relation $\widehat{\bfB} \widehat{\bfx}=\bfy$ is an overdetermined system of linear equations (OSLE) for the unknown $\widehat{\bfx}$. Let us assume that $\widehat{\bfB}$ has full column-rank. This holds w.p. $1$ when the rows of $\bfB$ are sampled from a continuous distribution. Hence, two cases are possible:
\begin{enumerate}
\item The resulting OSLE does not have a solution: In that case, the recovery algorithm neglects $\widehat{\bfB}$, because it is not compatible with the measurements $\bfy$. It selects a new matrix by selecting a new ordered set of $N$ distinct rows of $\bfA$.

\item The resulting OSLE has a solution: In that case, as we assume that $\widehat{\bfB}$ has full column-rank, this solution is the unique solution for the choice of the measurements $\widehat{\bfB}$. 
The algorithm outputs this candidate solution and terminates.
\end{enumerate}
Note that $\bfx$ itself is a candidate solution, which will be returned by the algorithm when $\widehat{\bfB}=\bfB$. 
It is easily seen that the recovery is successful w.p. $1$ if and only if every OSLE of the form $\widehat{\bfB}\widehat{\bfx}=\bfy$ either has no solution or its solution is equal to the initial signal $\bfx$. In the next section we prove that if the rows of $\bfA$ are randomly sampled from an arbitrary continuous distribution and $N \geq 2K$, then for all initial signals $\bfx\in \bR^K$, the brute-force algorithm successfully finds the initial signal $\bfx$ w.p. $1$. 
 

\section{Proof of the Main Result}\label{sec:proof}
In this section we present the proof of the main result of Theorem \ref{thm:main_theorem}. 
The main tool we use in our proof is the following property of polynomials.
For completeness, we provide a proof in the Appendix \ref{app:lem:zeroprob}.

\begin{lemma}\label{lem:zeroprob}
 Let $p(x_1,x_2, \dots, x_n)$ be a nonzero polynomial of the variables $x_1,\dots,x_n$ and let $X_1, X_2, \dots, X_n$ be independent (\textit{not necessarily identically distributed}) continuous random variables. Then $p(X_1,X_2, \dots,X_n) \neq 0$ w.p. $1$. \qed
\end{lemma}


The proof of the main result follows by carefully considering a number of different possibilities and applying this property of polynomials repeatedly.




The proof strategy is to argue that, w.p. $1$, for any $\widehat{\bfB}$ chosen at some stage of the recovery algorithm if a solution $\widehat \bfx$ exists then $\widehat \bfx = \bfx$.
We also need the following lemma which we will use frequently in the proof. 
The result is immediate from the invertibility of permutation matrices.
\begin{lemma}\label{lem:permokay}
For any $\widehat{\bfB}$ chosen at some stage of the recovery algorithm, a solution $\widehat \bfx$ satisfying $\widehat{\bfB} \widehat \bfx = \bfB \bfx$ exists if and only if ${\bf \Pi} \widehat{\bfB} \widehat \bfx = {\bf \Pi} \bfB \bfx$ for some permutation matrix ${\bf \Pi}$. \qed
\end{lemma}

Let $\bfC=[\bfB, \widehat{\bfB}]$ denote the $N\times 2K$ matrix obtained by concatenating the rows of $\bfB$ with the corresponding rows of $\widehat{\bfB}$. 
For most of the proof we will work with $\bfC$ instead of working explicitly with $\bfB$.
We need the following notation. 
For the chosen matrix $\widehat{\bfB}$ and the original matrix $\bfB$ 
we define a cycle as the \textit{longest} sequence $\bfv^1, \bfv^2, \dots, \bfv^n$ of the rows of $\bfA$, for some $n \geq 2$ with the following property:
\begin{enumerate}
\item $\bfv^1, \bfv^2, \dots, \bfv^{n-1}$ belong to the row set of $\bfB$,
\item $\bfv^2, \bfv^3, \dots, \bfv^{n}$ belong to the row set of $\widehat{\bfB}$,
\item for $k\leq n-1$, the row number of $\bfv^k$ in $\bfB$ is the same as the row number of $\bfv^{k+1}$ in $\widehat{\bfB}$,
\item all the elements of the sequence are different except for $\bfv^n$ which can be equal to $\bfv^1$.
\end{enumerate}

\noindent We call a cycle complete if $\bfv^1=\bfv^n$, otherwise we call it incomplete. 
With this definition of a cycle, it is clear that the rows of matrix $\bfC$ can be decomposed into a set of disjoint cycles.
Thus for given matrices $\bfB$ and $\widehat{\bfB}$, there is a unique cycle decomposition in terms of either complete or incomplete cycles. In particular, this decomposition partitions the rows of $\bfB$ and $\widehat{\bfB}$, namely, each row of $\bfB$ or $\widehat{\bfB}$ can be in one and only one complete or incomplete cycle. 

As we argued in Lemma \ref{lem:permokay}, for a given $\bfB$ and $\widehat \bfB$, the existence of a solution is unaffected by identical permutations of the rows of $\bfB$ and $\widehat \bfB$, and hence it follows that permuting the rows of $\bfC$ also does not affect the existence of a solution. 
We will now define a specific ordering of the rows of $\bfC$ that we will use in the proofs for ease of exposition.
For any $\widehat \bfB$ chosen during the algorithm, without affecting the existence of a solution, we can reorder the rows of the corresponding $\bfC$ so that the first $n-1$ rows of $\bfC$ consist of the rows involved in a cycle as  follows:
\begin{align}\label{cycle_shape}
\bfC=[\bfB, \widehat{\bfB}]=\left[ \begin{array}{cc} \bfv^1 & \bfv^2 \\ \bfv^2 & \bfv^3\\ \vdots & \vdots\\ \bfv^{n-1} & \bfv^n \\ \hline  \vdots & \vdots \end{array} \right ].
\end{align}
Note that in this configuration, the rows of $\bfB$ appear in the same order as that in which they appear in the cycle.
We can go one step further and rearrange the rows of $\bfC$ further until we attain a configuration in which the rows of $\bfB$ corresponding to each cycle appear in adjacent rows, and furthermore, these rows are arranged in the same order as in the cycle.
In addition, we ensure that the complete cycles appear before the incomplete cycles.
Furthermore, since all entries of $\bfA$ are drawn i.i.d., without loss of generality, we can assume that the rows of $\bfB$ are numbered $\bfv^1, \bfv^2, \ldots, \bfv^N$.
If the rows of $\bfC$ satisfy these properties, we say that $\bfC$ is in \textit{cycle-ordered form}.
The following example illustrates a $\bfC$ in cycle-ordered form.
\begin{align}\label{cycle_shape2}
\bfC=[\bfB, \widehat{\bfB}]=\left[ \begin{array}{cc} \bfv^1 & \bfv^2 \\ \bfv^2 & \bfv^3\\ \bfv^{3} & \bfv^1 \\ \hline  \bfv^4  & \bfv^5 \\ \  \bfv^5  & \bfv^4 \\ \hline  \bfv^6  & \bfv^7 \end{array} \right ].
\end{align}
In the example of \eqref{cycle_shape2} we have $M \geq 7$ and $N=6$.
In all there are three cycles, two complete cycles of length $3$ and $2$, and an incomplete cycle of length $1$.



The cycle decomposition allows us to split the proof into different cases depending on the number of cycles and complete cycles in the cycle representation of $\bfC$.
The different cases considered are shown in Table~\ref{proof_parts}.

We first consider the trivial case when $\bfx=0$.
If $\bfx=0$, the measurement vector is $\bfy=0$. We prove that in this case the recovery algorithm returns $\widehat{\bfx}=0$ w.p. $1$, where probability is taken with respect to the random construction of the rows of $\bfA$.
\begin{proposition}\label{prop:xzero}
If the desired signal $\bfx=0$, then the brute force algorithm returns the solution $\widehat \bfx =0$ w.p. $1$.
\end{proposition}

\begin{proof}
For $\bfx=0$, at each stage of the algorithm the recovery algorithm solves the OSLE $\widehat{\bfB} \widehat{\bfx}=\bfy=0$ to find a candidate solution for some candidate $\widehat{\bfB}$. 
Note that for $\bfy=0$, this equation has always the trivial solution $\widehat{\bfx}=0$, thus, to prove that there is no other solution, it is sufficient to prove that $\widehat{\bfB}$ has full column-rank. To prove this, we show that a $K\times K$ submatrix of $\widehat{\bfB}$ has nonzero determinant w.p. $1$. Without loss of generality, let $\widehat{\bfB}^1$ be a $K\times K$ matrix consisting of the first $K$ rows of $\widehat{\bfB}$. Note that the determinant of $\widehat{\bfB}^1$ is a nonzero polynomial of the $K^2$ components of $\widehat{\bfB}^1$. As the components of $\widehat{\bfB}^1$ are sampled independently from a continuous distribution, using Lemma \ref{lem:zeroprob}, the determinant of $\widehat{\bfB}^1$ is nonzero w.p. $1$, which implies that $\widehat{\bfB}$ has full column-rank w.p. $1$, thus, $\widehat{\bfx}=0$ is the unique solution for $\widehat{\bfB} \widehat{\bfx}=\bfy=0$.
\end{proof}
\vspace{2mm}

Now let us consider the more interesting case $\bfx\neq 0$. For this case, we use the number of complete cycles in the cycle decomposition in order to break up the proof as summarized in Table \ref{proof_parts}. We divide the proof into two parts:
\begin{enumerate}
\item If the number of complete cycles in $\bfC$ is greater than or equal to $K$, we prove that the OSLE $\widehat{\bfB}\widehat{\bfx}=\bfy$ either has no solution, in which case it is neglected by the recovery algorithm, or its solution is exactly equal to $\bfx$. In both cases, the brute force algorithm does not produce a wrong result for the given $\widehat{\bfB}$. 

\item If the number of complete cycles in $\bfC$  is less than $K$, we prove that the OSLE $\widehat{\bfB}\widehat{\bfx}=\bfy$ does not have a solution w.p. $1$.
\end{enumerate}
In both cases, we do not obtain a result different than $\bfx$. As there is at least one case $\widehat{\bfB}=\bfB$ for which the signal $\widehat{\bfx}=\bfx$ is correctly recovered, the brute force algorithm finds the initial signal $\bfx$ w.p. $1$.

\begin{table}[h]
\centering
\caption{Different parts of the proof}
\label{proof_parts}
\begin{tabular}{|c|c|c|c|c|}
\hline
Different Cases and Sub-cases & Proposition     & Result    & Check      \\
\hline
$\bfx=0$  & Proposition~\ref{prop:xzero}  & $\widehat{\bfx}=0=\bfx$ & \checkmark \\
\hline
\multirow{2}{*}{$\bfx\neq0$ and\myhash\,Complete Cycles in $\bfC$ $\geq K$}    & 
\multirow{2}{*}{Proposition \ref{complete>=K}} & No Solution    &  \checkmark          \\
&         & A Solution:  $\widehat{\bfx}=\bfx$    &  \checkmark          \\
\hline
$\bfx \neq 0$ and\myhash\,Complete Cycles in $\bfC$ $\leq K-1$  
&  Proposition  \ref{complete<K}  & No Solution & \checkmark          \\
\hline
\end{tabular}
\end{table}


\subsection{The number of complete cycles in $\bfC$ is greater than or equal to $K$}
For this case, we prove that when $\bfx\neq 0$, the recovery algorithm either does not have a solution or it returns a solution equal to $\bfx$. 


\begin{proposition}\label{complete>=K}
Let $\bfB$ and $\widehat{\bfB}$ be as before and assume that the number of complete cycles in the cycle representation of $\bfC=[\bfB, \widehat{\bfB}]$ is greater than or equal to $K$. Then, w.p. $1$, for any $\bfx \in \bR^K \setminus \{0\}$, the solution set consisting of all $\widehat \bfx$ satisfying the OSLE $\widehat{\bfB}\widehat{\bfx}=\bfy=\bfB \bfx$ is either the empty set or the singleton $\{\bfx\}$.  
\end{proposition}
\begin{proof}
We decompose the set of nonzero signals $\bfx$ into two sets: $\clX_0$ and $\clX_1$. The set $\clX_0$ consists of those $\bfx$ for which the OSLE $\widehat{\bfB} \widehat{\bfx}=\bfy=\bfB \bfx$ does not have a solution, thus, for $\clX_0$, the result immediately follows. Let us consider those $\bfx \in \clX_1=\bR^K - (\clX_0 \cup \{0\})$ for which the OSLE does have a solution. We will show that w.p. $1$, for all $\bfx \in \clX_1$, the resulting solution $\widehat{\bfx}$ is equal to $\bfx$. 

Let the number of complete cycles in the cycle decomposition of $\bfB$ and $\widehat{\bfB}$ be $m$. 
Let $\bfv^1,\bfv^2, \dots, \bfv^{n}, \bfv^1$ be any complete cycle of length $n \geq 1$. 
As we argued in \eqref{cycle_shape} and \eqref{cycle_shape2}, without loss of generality we can reorder the rows of the concatenated matrix $\bfC$ in cycle-ordered form such that the rows representing the cycle appear first as shown in \eqref{cycle_shape}.
From $\widehat{\bfB}\widehat{\bfx}=\bfy=\bfB \bfx$, and the cycle representation shown in \eqref{cycle_shape}, it follows that 
\begin{align}
\inp{\bfv^i}{\bfx}=\inp{\bfv^{(i \text{ mod } n) + 1}}{\widehat{\bfx}}, i=1, \dots, n.
\end{align}
Summing up both sides of the above equation over $i=1,2,\dots, n$, we obtain $\inp{\bfw^1}{\bfx-\widehat{\bfx}}=0$, where $\bfw^1 \triangleq \sum_{i=1}^{n} \bfv^i$. Performing the same steps for the other cycles, we obtain vectors $\bfw^\ell$, $\ell=1,2, \dots, m$, with 
\begin{equation}\label{eqn:welldefn}
\inp{\bfw^\ell}{\bfx-\widehat{\bfx}}=0, 
\end{equation}
where $\bfw^\ell$ is the sum of the vectors within the $\ell$-th cycle. As complete cycles consist of disjoint subset of the rows of $\bfA$ and as these vectors are generated independently from one another, it is clear that $\bfw^\ell$ are independent vectors. Moreover, the components of each $\bfw^\ell$ are i.i.d. with a continuous distribution $f_\ell=\underbrace{f \star f \star \dots \star f}_{n_\ell}$, where $n_\ell$ is the length of the complete cycle $\ell$, and where $\star$ denotes the convolution operator. 

Let $\bfG$ be an $m \times K$ matrix whose $\ell$-th row is given by $\bfw^\ell$. By \eqref{eqn:welldefn} we have $\bfG (\bfx-\widehat{\bfx})=0$. The probability that for every $\bfx \in \clX_1$, the solution set for $\widehat{\bfB}\widehat{\bfx}=\bfy=\bfB \bfx$ is either empty or the singleton $\{\bfx\}$, is lower-bounded by the probability that $\bfG$ has full column-rank. We prove that the latter probability is equal to $1$, which implies the desired result. 

To prove this, consider the $K \times K$ matrix $\widetilde \bfG$ consisting of the first $K$ rows of $\bfG$. Note that the determinant of $\widetilde \bfG$ is a nonzero polynomial of its components. As the components of $\bfw^\ell$ are independent with a continuous distribution, using Lemma \ref{lem:zeroprob}, it follows that, w.p. $1$, $\det(\widetilde \bfG) \neq 0$ and thus $\widetilde \bfG$ is invertible. This implies that $\bfG$ is indeed full-rank w.p. $1$ thus completing the proof.
\end{proof}

\subsection{The number of complete cycles in $\bfC$ is less than $K$}\label{sec:tot_comp<K}
In this section, we consider the last case where $\bfx\neq 0$ and the number of complete cycles in the cycle decomposition of $\bfB$ and $\widehat{\bfB}$ is less than $K$ as shown in the third row of Table \ref{proof_parts}. In this case, the recovery algorithm solves the OSLE $\widehat{\bfB} \widehat{\bfx}=\bfy=\bfB \bfx$ to find a candidate solution for $\widehat{\bfx}$. As before let $\bfC=[\bfB, \widehat{\bfB}]$ be the concatenation of $\bfB$ and $\widehat \bfB$. Also let $\bfu=[-\bfx; \widehat{\bfx}]$. Then the OSLE can be equivalently written  as $\bfC \bfu=0$. We claim that if  $N \geq 2K$, then $\bfC$ has full column-rank, which implies that the only possible solution is $\bfu=[-\bfx; \widehat{\bfx}]=0$. As we assume that $\bfx\neq 0$, it follows that the OSLE $\widehat{\bfB} \widehat{\bfx}=\bfy=\bfB \bfx$ cannot have a solution,  and thus the corresponding choice of $\widehat \bfB$ is ignored by the recovery algorithm. 

As before, let us assume, without loss of generality, that the rows of $\bfC$ are in cycle-ordered form, i.e., the rows corresponding to the complete cycles in $\bfC$ are placed first followed by the rows corresponding to the incomplete cycles as illustrated in \eqref{cycle_shape2}.
For the purpose of this proof we introduce some new notation. 
Let $\bfD$ be the matrix comprising the first $2K$ rows of $\bfC$ in cycle-ordered form.
We define cycles, complete cycles, and incomplete cycles for $\bfD$ in the same manner as we defined them for $\bfC$.
Thus all the cycles appearing in the first $2K$ rows of $\bfC$ appear in $\bfD$ as well with the possible exception that {the last cycle in $\bfD$ might be the truncation of the corresponding cycle in $\bfC$}.
From our assumptions on the arrangement of the rows of $\bfC$, it is not difficult to check that the number of complete cycles in $\bfD$ {is less than or equal to the number of complete cycles in $\bfC$ and thus it is still strictly less than $K$. But the total number of cycles, counting both complete and incomplete ones, might be less than, equal to or greater than $K$. 
We will prove that, as long as the number of complete cycles in $\bfD$ is less than $K$, w.p. $1$, $\det(\bfD)\neq 0$, which implies that the initial matrix $\bfC$ is full-rank.


We divide the proof into two parts based on the total number of cycles (complete or incomplete) in $\bfD$. 
The first part addresses the case where the total number of cycles in $\bfD$ is greater than or equal to $K$ and the second part considers the case where the number of cycles is less than $K$.
In both cases, we pursue a standard procedure to prove that $\det(\bfD) \neq 0$ w.p. $1$. More precisely, let $\bfv^1, \bfv^2, \dots, \bfv^t$ be the total number of rows of the initial matrix $\bfA$ that exists in $\bfD$. Note that $\det(\bfD)$ is a polynomial of the entries of these variables. 
Thus if we can find a specific assignment to these variables such that the $\det(\bfD)$ is nonzero then it follows that the polynomial is a non-zero polynomial.
As $\bfv^i$, $i \in [t]$, are randomly sampled according to a continuous distribution, by applying Lemma \ref{lem:zeroprob}, we immediately obtain that $\det(\bfD)\neq 0$ w.p. $1$. Thus, essentially the main idea is to find a suitable assignment with $\det(\bfD)\neq 0$. 
In both cases, we will use induction on the dimension of the signal $K$ to show that for every $K$ such an assignment exists. 


\vspace{3mm}
\subsubsection{Total number of cycles in $\bfD$ is greater than or equal to $K$}
We first prove the case in which the number of complete cycles is exactly $K-1$.

\begin{proposition}\label{total>=K_complete=K-1}
Assume that the number of complete cycles in $\bfD$ is exactly $K-1$ and the total number of cycles (complete or incomplete) is  greater than or equal to $K$. Then, there is an assignment to the variables for which $\det(\bfD) \neq 0$.  
\end{proposition}

\begin{proof}
We use induction on $K$. 
For $K=1$, {(all the rows of $\bfA$ are scalars)} and so we use $v_i$ to denote the $i$-th entry of $\bfA$. 
In this case, there should be no complete cycles and the total number of cycles must be greater than or equal to one. 
Thus only the following two cases can happen:
\begin{align}\label{total>=K:eq1}
\bfD^1=\left [ \begin{array}{cc} v_1 & v_2 \\ v_3 & v_4 \end{array} \right ], \ \bfD^2=\left [ \begin{array}{cc} v_1 & v_2 \\ v_2 & v_3 \end{array} \right ],
\end{align}
where $\bfD^1$ and $\bfD^2$ have two and one incomplete cycles respectively. We can simply check that $p_1=\det(\bfD^1)=v_1 v_4 -v_2 v_3$ and $p_2=\det(\bfD^2)=v_1v_3 -v_2^2$ are both nonzero polynomials, thus, a suitable assignment trivially exists.

Now we assume that the induction hypothesis holds for $K$ and we extend it to $K+1$. Consider a $2(K+1)\times 2(K+1)$ matrix $\bfD$ such that $\bfD$ contains exactly $K$ complete cycles and that the total number of cycles in $\bfD$ is greater than or equal to $K+1$. The proof requires checking several cases that have been listed below:

\begin{enumerate}
\item There is a complete cycle of length $2$ in $\bfD$.
\item There are no complete cycles of length $2$ but there is a complete cycle of length more than $2$. 
\item All the complete cycles have length $1$. This contains the following two sub-cases:
\begin{itemize}
\item [3-1)] There is an incomplete cycle of length $2$ or more.
\item [3-2)] All the incomplete cycles have length $1$.
\end{itemize}
\end{enumerate}  
  
\noindent 1) If there is a complete cycle of length $2$, the matrix $\bfD$ will be as in Equation \eqref{total>=K:eq2}.

\begin{align}\label{total>=K:eq2}
\bfD=\left [ \begin{array}{ccccccccc} \bfv_1^1 & \bfv^1_2 & \cdots & \bfv^1_{K+1} &\Big |& \bfv^2_1 & \bfv^2_2 & \cdots & \bfv^2_{K+1}\\
\bfv^2_1&\bfv^2_2 & \cdots & \bfv^2_{K+1} &\Big | &\bfv_1^1&\bfv^1_2 & \cdots & \bfv^1_{K+1}\\
\hline 
& & \vdots & & \Big | & &  & \vdots & \\
\end{array} \right ]
\end{align}
As an assignment, we set $\bfv^1_{K+1}=1$ and set all the remaining components of $\bfv^1, \bfv^2$ equal to zero. By expanding the determinant with respect to the first and the second row, we can see that $\det(\bfD)=\pm \det(\widetilde{\bfD})$ where $\widetilde{\bfD}$ is a $2K \times 2K$ matrix obtained by removing the first and the second row of $\bfD$ along with the columns $K+1$ and $2(K+1)$. It is not difficult to see that $\widetilde{\bfD}$ is a $2K\times 2K$ matrix and has $K-1$ complete cycles because one complete cycle was removed by the determinant expansion. Also, its total number of cycles is greater or equal to $K$, thus, using the induction hypothesis, there is an assignment to the remaining variables in $\widetilde{\bfD}$ with $\det(\widetilde{\bfD}) \neq 0$, which along with $\bfv^1_{K+1}=1$ and $\bfv^i_{K+1}=0$ for all $i > 2$, gives a suitable assignment to the matrix $\bfD$.  

\noindent 2) If there is no complete cycle of length $2$ but there is at least one complete cycle of length more than $2$, we follow a similar procedure. In this case, the matrix $\bfD$ can be represented as in \eqref{total>=K_one_cyclegt2}.

\begin{align}\label{total>=K_one_cyclegt2}
\bfD=\left [ \begin{array}{ccccccccc} \bfv_1^1 & \bfv^1_2 & \cdots & \bfv^1_{K+1} &\Big |& \bfv^2_1 & \bfv^2_2 & \cdots & \bfv^2_{K+1}\\
\bfv^2_1&\bfv^2_2 & \cdots & \bfv^2_{K+1} &\Big | &\bfv^3_1&\bfv^3_2 & \cdots & \bfv^3_{K+1}\\
& & \vdots & & \Big | & &  & \vdots & \\
& & \vdots & & \Big | &\bfv_1^1&\bfv^1_2 & \cdots & \bfv^1_{K+1}\\
\hline 
& & \vdots & & \Big | & &  & \vdots & \\
\end{array} \right ]
\end{align}

As an assignment, we set $\bfv^1_{K+1}$ equal to $\lambda$ and all the remaining components of $\bfv^1$ equal to zero. 
We also set $\bfv^i_{K+1} = 0$ for all $i >1$.
We keep $\lambda$ as a parameter and specify its value later. By expanding the determinant with respect to the first and the last row in the cycle, we can check that $\det(\bfD)$ is a quadratic function of $\lambda$, where the coefficient of $\lambda^2$ is given by $\pm \det(\widetilde{\bfD})$, where $\widetilde{\bfD}$ is the $2K\times 2K$ matrix obtained after removing the rows corresponding to $\bfv^1$ in the cycle (the first and the last row of the cycle) along with the columns $K+1$ and $2(K+1)$. It is not difficult to see that $\widetilde{\bfD}$ is as in \eqref{total>=K:eq3}, where $\widetilde{\bfv}^i \in \bR^K$  denotes a $K$-dimensional vector obtained by removing the last component of $\bfv^i$. 

\begin{align}\label{total>=K:eq3}
\widetilde{\bfD}=\left [ \begin{array}{ccc} 
\widetilde{\bfv}^2 & \Big | & \widetilde{\bfv}^3 \\
\widetilde{\bfv}^3 & \Big |& \widetilde{\bfv}^4 \\
\vdots & \Big | & \vdots \\
\hline 
 \vdots &\Big | & \vdots \\
\end{array} \right ]
\end{align}

We still need to assign values to $\widetilde{\bfv}^i$. Note that $\widetilde{\bfD}$ satisfies the induction hypothesis, i.e., it is a $2K\times 2K$ matrix with exactly $K-1$ complete cycles (one complete cycles was removed by the truncation) with total number of cycles  greater than $K$. Thus, there is an assignment to the variables with a nonzero $\det(\widetilde{\bfD})$. 
This implies that given the assignments thus far, $\det(\bfD)$ is a quadratic function of $\lambda$ such that the coefficient of $\lambda^2$ is non-zero. 
Therefore, we can find at least one value of $\lambda=\lambda^*$ such that $\det(\bfD)\neq 0$, which is the desired assignment for $\bfD$.

\noindent 3) Finally, we consider the case where all the complete cycles have length $1$. Then the matrix $\bfD$ has $K$ complete cycles of length $1$, where each complete cycle occupies one row of the matrix $\bfD$. Consider the $K+2$ remaining rows of $\bfD$. Again we need to consider two different cases.

\noindent 3-1) Assume that there is an incomplete cycle of length $2$ or more. This has been shown in \eqref{total>=K:ic_cycle>2},  where $\bfv^1$ denotes a specific complete cycle and where, for ease of illustration, we assume that the incomplete cycle has length $2$ and, without loss of generality, it is adjacent to the complete first complete cycle of length $1$. 
The proof can be extended to an incomplete cycle of length greater than $2$. 

\begin{align}\label{total>=K:ic_cycle>2}
\widetilde{\bfD}=\left [ \begin{array}{ccccccccc}
 \bfv_1^1 & \bfv^1_2 & \cdots & \bfv^1_{K+1} &\Big |& \bfv_1^1 & \bfv^1_2 & \cdots & \bfv^1_{K+1}\\
\hline
\bfv^2_1&\bfv^2_2 & \cdots & \bfv^2_{K+1} &\Big | &\bfv^3_1&\bfv^3_2 & \cdots & \bfv^3_{K+1}\\
\bfv^3_1&\bfv^3_2 & \cdots & \bfv^3_{K+1} &\Big | &\bfv^4_1&\bfv^4_2 & \cdots & \bfv^4_{K+1} \\
\hline 
& & \vdots & & \Big | & &  & \vdots & \\
\end{array} \right ]
\end{align}

We set $\bfv^1_{K+1}$ and $\bfv^4_{K+1}$ equal to $\lambda$. We also set all the other components of $\bfv^1$ and $\bfv^4$ equal to zero. By expanding the determinant with respect to the first and the third row of $\bfD$, it can be checked that $\det(\bfD)$ is again quadratic in terms of $\lambda$, where the coefficient of $\lambda^2$ is given by the determinant of $\widetilde{\bfD}$ where $\widetilde{\bfD}$ is given by \eqref{total>=K:eq4}, 
\begin{align}\label{total>=K:eq4}
\widetilde{\bfD}=\left [ \begin{array}{ccc} 
\widetilde{\bfv}^2 & \Big | & \widetilde{\bfv}^3 \\
\hline 
 \vdots &\Big | & \vdots \\
\end{array} \right ]
\end{align}
where $\widetilde{\bfv}^i \in \bR^K$  again denotes a $K$-dimensional vector obtained by removing the last component of $\bfv^i$. 
Note that $\widetilde{\bfD}$ satisfies the induction hypothesis because it is a $2K\times 2K$ matrix with $K-1$ complete cycles (one complete cycle was removed by truncation) whose total number of cycles is greater or equal to $K$. We can now build a suitable assignment for $\lambda$ and $\widetilde{\bfD}$ {by using the induction hypothesis and} following the same procedure outlined after \eqref{total>=K:eq3}.
Also note that if the incomplete cycle has length more than $2$, then we obtain a matrix $\widetilde{\bfD}$ similar to that in \eqref{total>=K:eq4} but with more than $1$ row in its modified incomplete cycle, and a similar argument allows us to determine a valid assignment to the variables.

\noindent 3-2) The only case that remains to be checked is when all the incomplete cycles are of length $1$. {Recall that all the complete cycles are also of length $1$.} In this case, there are at least two incomplete cycles $\bfv^1, \bfv^2$ and $\bfv^3, \bfv^4$. As we are looking for a suitable assignment with $\det(\bfD)\neq 0$, we can always assign equal values to $\bfv^2$ and $\bfv^3$. This essentially implies that we can merge these two incomplete cycles to build the incomplete cycle $\bfv^1, \bfv^2=\bfv^3, \bfv^4$. By this modification, we obtain an incomplete cycle of length $2$. Note that {as all the cycles in this case have length $1$,} the total number of cycles after this modification is $2(K+1)-1=2K+1$, which is still  greater or equal to $K+1$ for all $K=1,2,\dots$. As we have already argued in part 3-1), i.e., the case for which there is at least one incomplete cycle of length $2$ or more, that a suitable assignment is possible, the proof is complete.
\end{proof}

 Now we can prove the more general case, in which the total number of complete cycles in $\bfD$ is less than $K$.

\begin{proposition}\label{total>=K_complete<K}
Assume that the number of complete cycles in $\bfD$ is less than $K$ and the total number of cycles (complete or incomplete) is  greater than or equal to $K$. Then, there is assignment to the variables for which $\det(\bfD) \neq 0$.  
\end{proposition}

\begin{proof}
We argue that, in this case, without loss of generality, we can assume that the number of complete cycles is exactly $K-1$, thus, using Proposition \ref{total>=K_complete=K-1}, we obtain the proof. To explain this, suppose that $\bfv^1, \bfv^2, \dots, \bfv^\ell$ is an incomplete cycle. As we are essentially looking for a suitable assignment, we can always assign equal values to the first and the last vector in this cycle, i.e., we can set $\bfv^1=\bfv^\ell$. With this assignment, an incomplete cycle can be treated like a complete cycle. Consequently, if the number of complete cycles is less than $K-1$, we can always convert some of the incomplete cycles into complete ones to keep the number of complete cycles equal to $K-1$. Moreover, the total number of cycles after conversion is still  greater or equal to $K$. Thus, using Proposition \ref{total>=K_complete=K-1}, we obtain the proof. 
\end{proof}

\vspace{3mm}
\subsubsection{Total number of cycles in $\bfD$ is less than $K$}
First note that in this case, without loss of generality we can assume that all the cycles are complete. More precisely, consider an incomplete cycle $\bfv^1, \bfv^2, \dots, \bfv^\ell$. As we are looking for a suitable assignment for these variables, we can always assign equal values to $\bfv^1$ and $\bfv^\ell$, which implies that we can treat this incomplete cycle like a complete one. Hence, we can always assume that all the cycles are complete. We need to prove the following proposition. 

\begin{proposition}\label{total<K_complete<K}
Let $\bfD$ be a $2K \times 2K$ matrix as before. Assume that all the cycles of $\bfD$ are complete, and the number of these complete cycles is strictly less than $K$. Then, there is an assignment to the variables in $\bfD$ with $\det(\bfD)\neq 0$. 
\end{proposition}

\begin{proof}
We prove this result using induction on $K$.
For $K=1$, the number of cycles should be $0$, thus, we need to start the induction from $K=2$.  For $K=2$, we can have one complete cycle as in  \eqref{total<K:K=2}, where $\bfv^\ell$ are vectors in $\bR^2$:
\begin{align}\label{total<K:K=2}
\bfD=\left[ \begin{array}{cc} \bfv^1 & \bfv^2 \\ \bfv^2 & \bfv^3\\ \bfv^3 & \bfv^4\\ \bfv^4 & \bfv^1 \end{array} \right ].
\end{align}
Consider the following assignment to the variables: $\bfv^1=[1, 0]$, $\bfv^3=[0,1]$ and $\bfv^2=\bfv^4=[0, 0]$. We can simply check that for this assignment, matrix $\bfD$ is a matrix with exactly one $1$ in each row and each column, and zeros elsewhere. Thus, it is a permutation matrix and $\det(\bfD) \neq 0$.

%
%

Now, assume that the result holds for some {$K\geq 2$} and consider the statement for $K+1$. Let $\bfD$ be a $2(K+1)\times 2(K+1)$ matrix with less than $K$ cycles such that all the cycles are complete. There should be a complete cycle of length at least $3$ in $\bfD$ otherwise the total number of the rows of the matrix must be less than or equal to $2K$, which lead to{ a contradiction}. We now provide the proof argument for the case in which the complete cycle has length $3$ 
as shown in \eqref{total<K:cycle3}, where for ease of illustration we assume that this cycle is the first cycle in $\bfD$. The same idea can be modified for the case in which the complete cycle has a length greater than $3$.

\begin{align}\label{total<K:cycle3}
\bfD=\left [ \begin{array}{ccccccccc} \bfv_1^1 & \bfv^1_2 & \cdots & \bfv^1_{K+1} &\Big |& \bfv^2_1 & \bfv^2_2 & \cdots & \bfv^2_{K+1}\\
\bfv^2_1&\bfv^2_2 & \cdots & \bfv^2_{K+1} &\Big | &\bfv^3_1&\bfv^3_2 & \cdots & \bfv^3_{K+1}\\
\bfv^3_1&\bfv^3_2 & \cdots & \bfv^3_{K+1} &\Big | &\bfv_1^1&\bfv^1_2 & \cdots & \bfv^1_{K+1}\\
\hline 
& & \vdots & & \Big | & & & \vdots  & \\
\end{array} \right ]
\end{align}

We assign value $\lambda$ to $\bfv^1_{K+1}$, where $\lambda$ will be specified later. By expanding the determinant with respect to the first and the third row of $\bfD$, it is seen that $\det(\bfD)$ is a quadratic function of $\lambda$, where the coefficient of $\lambda^2$ is given by the determinant of $\widetilde{\bfD}$ given by \eqref{total<K:cycle3_truncated}.
\begin{align}\label{total<K:cycle3_truncated}
\widetilde{\bfD}=\left [ \begin{array}{ccc}
\widetilde{\bfv}^2 &\Big | &\widetilde{\bfv}^3 \\
\hline 
\vdots &\Big | &\vdots
\end{array} \right ] , 
\end{align}
where $\widetilde{\bfv}^i$ is $\bfv^i$ after removing the last component. Note that $\widetilde{\bfD}$ is a $2K\times 2K$ matrix whose number of complete cycles is less than $K$. However the number of cycles in $\bfD$ and $\widetilde{\bfD}$ is the same because one incomplete cycle is created after the truncation. 

We already know that the total number of cycles in the initial matrix $\bfD$ is less than $K+1$. Therefore, two cases can happen. If the total number of cycles in $\bfD$ is less than $K$, then $\widetilde{\bfD}$ satisfies the induction hypothesis since it will have less than $K$ cycles, where the newly created incomplete cycle can be treated as a complete one. Consequently, we can find an assignment to the variables in $\widetilde{\bfD}$ with a nonzero $\det(\widetilde{\bfD})$. We also assign zero to $\bfv^i_{K+1}$ for $i>1$. Thus, there exists at least one $\lambda=\lambda^*$ such that $\det(\bfD) \neq 0$. Assigning this value to $\bfv^1_{K+1}$, together with the values assigned to the remaining variables yields a suitable assignment for $\bfD$ that leads to $\det(\bfD) \neq 0$. 


Finally, if the number of cycles in $\bfD$ is exactly $K$, then $\widetilde{\bfD}$ 
{will have $K-1$ complete cycles and one incomplete one.} Using Proposition \ref{total>=K_complete<K}, there must be an assignment to $\widetilde{\bfD}$ with $\det(\widetilde{\bfD}) \neq 0$. Following similarly to the previous case, we can find a suitable value $\lambda=\lambda^*$ which along with $\widetilde{\bfD}$ gives a suitable assignment to the initial matrix $\bfD$. This completes the proof. 
\end{proof}

Putting together the results of Propositions \ref{total>=K_complete<K} and \ref{total<K_complete<K}, we can now argue that if the total number of complete cycles in $\bfC$ is less than $K$ then for any non-zero $\bfx$, the recovery algorithm cannot produce a solution.

\begin{proposition}\label{complete<K}
Let $\bfB$ and $\widehat{\bfB}$ be as before and assume that the number of complete cycles in the cycle representation of $\bfC=[\bfB, \widehat{\bfB}]$ is less than $K$. Then, w.p. $1$, for any $\bfx \in \bR^K \setminus \{0\}$, there does not exist a solution $\widehat{\bfx}$ to the OSLE $\widehat{\bfB}\widehat{\bfx}=\bfy=\bfB \bfx$.  
\end{proposition}
\begin{proof}
From the results of Propositions \ref{total>=K_complete<K} and \ref{total<K_complete<K} it follows that if the number of complete cycles in the cycle representation of $\bfC=[\bfB, \widehat{\bfB}]$ is less than $K$, then there is an assignment to the variables in $\bfD$ such that $\det(\bfD)\neq 0$.
Thus, if the entries of $\bfA$ are drawn i.i.d. from a continuous distribution, then by applying Lemma \ref{lem:zeroprob}, we immediately obtain that $\det(\bfD)\neq 0$ w.p. $1$. 
As a result the null space of $\bfD$ and $\bfC$ are equal to the singleton $\{0\}$ w.p. $1$.
{Thus if $\widehat{\bfx}$ is a solution to the OSLE $\widehat{\bfB}\widehat{\bfx}=\bfy=\bfB \bfx$, then, $\bfu=[-\bfx;\widehat{\bfx}]$ must lie in the null-space of $\bfC$, which forces $\bfx$ to be $0$, which is a contradiction since we assumed that $\bfx\neq 0$. Hence, OSLE $\widehat{\bfB}\widehat{\bfx}=\bfy=\bfB \bfx$ can not have any solution.}
\end{proof}

Combining the results of Propositions \ref{prop:xzero}, \ref{complete>=K}, and \ref{complete<K}, we conclude that if $N \geq 2K$, then for any $\bfx \in \bR^K$ any solution to the OSLE $\widehat{\bfB}\widehat{\bfx}=\bfy=\bfB \bfx$ satisfies $\widehat{\bfx} = \bfx$.
This completes the proof of Theorem \ref{thm:main_theorem}.

\section{Converse result for the minimum number of measurements}\label{sec:converse}
In this section, we prove that if\footnote{In the degenerate case of $K=1$, recovery is possible with just $N=1$ measurement as there is no ambiguity in the ordering of the measurements.} $K >1$, then $2K$ is the minimum number of random measurements required for universal recovery of all signals in $\bR^K$. 
To show the result for $K>1$, we prove that if $N<2K$, then w.p. $1$, for any realization of $\bfA$, there is at least one signal $\bfx\neq 0$ for which the brute force algorithm of Section \ref{sec:rec_alg} produces
an incorrect solution. We consider the simple case $M=N$, in which case $\bfB$ is obtained by simply permuting the rows of $\bfA$. From the description of the algorithm in Section \ref{sec:rec_alg}, it is immediate that the converse result for this case provides a lower bound also on the required number of measurements in the more general case $M>N$.

As in Section \ref{sec:intro}, let $\bfB$ be the $N\times K$ matrix of measurements and let $\bfy=\bfB \bfx$ be the $N\times 1$ vector of measurements. As described earlier, the recovery algorithm considers the set of all {$N\times N$} permutation matrices ${\bf \Pi}$ and seeks a solution $\widehat{\bfx}$ to the linear equation $\bfy={\bf \Pi} \bfB \widehat{\bfx}$. 
Thus if we find a signal $\widehat{\bfx}$ with $\widehat{\bfx}\neq \bfx$ and a fixed permutation ${\bf \Pi}$ such that $\bfB \bfx={\bf \Pi} \bfB \widehat{\bfx}$, then it follows that the recovery algorithm could fail. 
In the rest of this section we prove the following result which is a converse to Theorem \ref{thm:main_theorem}.

\begin{theorem}\label{thm:converse}
For $N<2K$ and $K > 1$ let $\bfB$ be an $N\times K$ matrix whose components are drawn i.i.d. at random from a continuous distribution. Then, w.p. $1$, the problem of recovering $\bfx$ from the unordered entries of $\bfy$ does not in general admit a unique solution, i.e., w.p. 1, there exists a permutation matrix ${\bf \Pi}$ and vectors $\bfx, \widehat{\bfx}$ such that $\bfx \neq \widehat{\bfx}$  and $\bfB \bfx = {\bf \Pi} \bfB \widehat \bfx$. \qed
\end{theorem}


In analogy to the property of \eqref{eqn:nullspace}, this result can be equivalently stated as a property of random matrices. 
In essence for the random matrix $\bfB$ and for the given choice of parameters, there exists, w.p. 1, a permutation matrix $\Pi$ such that 
\[
\clN([\bfB, \bf \Pi \bfB]) \not \subset \clN([\bfI_K, \bfI_K]).
\]

We prove the converse result for even and odd values of $N<2K$ separately.
We begin with a result for a special case.

\begin{proposition}\label{converse_even_1}
Let $\bfB$ be an $N\times K$ matrix with $N=2K-2$ whose components are drawn i.i.d. at random from a continuous distribution. Then, w.p. $1$, 
there exists a permutation matrix ${\bf \Pi}$ and vectors $\bfx, \widehat{\bfx}$ such that $\bfx \neq \widehat{\bfx}$  and $\bfB \bfx = {\bf \Pi} \bfB \widehat \bfx$. 
\end{proposition}

\begin{proof}
For $K=1$, the result is trivial, thus, we assume that $K\geq 2$. Consider the matrix $\bfG=[\bfB, {\bf \Pi} \bfB]$ and choose a permutation ${\bf \Pi}$ such that  $\bfG$ has a cycle representation with less than $K-1$ cycles, e.g., consider the permutation ${\bf \Pi}$ with only one cycle given by ${\bf \Pi}_{i,j}=1$ only for those $(i,j)\in\{(1,2),(2,3), \dots, (N-1,N), (N,1)\}$. Note that since $M=N$, all the cycles will be necessarily complete by our definition. 


Let $\bfg^1$ and $\bfg^2$ denote the $K$-th and $2K$-th columns of $\bfG$.
Clearly $\bfg^2={\bf \Pi} \bfg^1$. 
Let $\widetilde{\bfG}$ be the $(2K-2)\times (2K-2)$ matrix obtained by removing $\bfg^1$ and $\bfg^2$ from $\bfG$. 
It is not difficult to check that matrix $\widetilde{\bfG}$ has still less than $K-1$ cycles as removing these two columns  does not change the number of cycles. 
As all the vectors are randomly sampled from a continuous distribution, $\|\bfg^2\|>0$ w.p. $1$. Moreover, since $\widetilde{\bfG}$ has only complete cycles and their count is less than $K-1$, from Proposition \ref{complete<K}, it must be invertible w.p.  $1$. 

Let $\bfx, \widehat{\bfx} \in \bR^K$. 
We introduce some notation.
Let $\underline{\bfx}, \widehat{\underline{\bfx}}$ denote $(K-1)$-dimensional vectors consisting of the first $K-1$ components of $\bfx$ and $\widehat{\bfx}$, and $\bfx_K$ and $\widehat{\bfx}_K$ denote the last component of $\bfx$ and $\widehat{\bfx}$ respectively, so that we have $\bfx=[\underline{\bfx}\,; \bfx_K]$ and $\widehat{\bfx}=[\widehat{\underline{\bfx}}\,; \widehat{\bfx}_K]$.
Let us define $\bfz=[-\bfx\, ; \widehat{\bfx}]$ and $\widetilde{\bfz}=[-\underline{\bfx}\,; \widehat{\underline{\bfx}}]$. Then we have  
\begin{align}
\bfG \bfz= \widetilde{\bfG} \widetilde{\bfz} + (\widehat{\bfx}_K \bfg^2- \bfx_K \bfg^1).
\end{align}
Clearly, as $\widetilde{\bfG}$ is invertible, and $\|\bfg^2\| > 0$  if we set $\widehat{\bfx}_K=1$ and $\bfx_K=0$, we can always find a nonzero solution for $\widetilde{\bfz}$, such that $\bfG \bfz = 0$.
This choice of $\widetilde{\bfz}$ gives a corresponding choice of values for $\underline{\bfx}$ and $\widehat{\underline{\bfx}}$ such that $\bfG \bfz = 0$.
As a result we can build a signal $\bfx=[\underline{\bfx}\,; 0]$ and an estimate $\widehat{\bfx}=[\widehat{\underline{\bfx}}\,; 1]$ that satisfy $\bfB\bfx={\bf \Pi} \bfB \widehat{\bfx}$. But then $\|\widehat{\bfx} - \bfx\| \geq |\widehat{\bfx}_K-\bfx_K|=1>0$ and thus $\widehat{\bfx} \neq \bfx$. This completes the proof.
\end{proof}

\begin{proposition}\label{prop:converse_even_2}
Let $\bfB$ be an $N\times K$ matrix where $N$ is an even number less than $2K$. Suppose that the components of $\bfB$ are drawn i.i.d. randomly from a continuous distribution. Then, w.p. $1$, the problem does not in general admit a unique solution, i.e., w.p. 1, there exists a permutation matrix ${\bf \Pi}$ and vectors $\bfx, \widehat{\bfx}$ such that $\bfx \neq \widehat{\bfx}$ and $\bfB \bfx = {\bf \Pi} \bfB \widehat{\bfx}$. 
\end{proposition}
\begin{proof}

For $K=1$ the result is trivial, thus, we assume that $K\geq 2$.
The proof simply follows from the proof of Proposition \ref{converse_even_1}. 
Suppose $N=2K-2r$ for some $1\leq r \leq K-1$. For $r=1$ the result follows from Proposition \ref{converse_even_1}. 
For $r\geq 2$, we can write $\bfB$ as $\bfB = \left[ \mathbf{\widetilde B}^1, \mathbf{\widetilde B}^2 \right]$ where $\mathbf{\widetilde B}^1$ is of dimension $N \times (K-r+1)$.
We know from Proposition \ref{converse_even_1} that, w.p.  $1$, there exists a permutation $\bf{\Pi}$ and distinct $(K-r+1)$-dimensional signals $ \underline{\bfx}$ and $ \underline{\widehat {\bfx}}$ such that $\mathbf{\widetilde B}^1 \underline{\bfx} = {\bf \Pi}\mathbf{\widetilde B}^1 \underline{\widehat {\bfx}} $.
Thus if we choose $\bfx$ such that the first $K-r+1$ entries of $\bfx$ equal $\underline{\bfx} $ and  $\widehat \bfx$ such that the first $K-r+1$ entries of $\widehat \bfx$ equals $\underline{\widehat {\bfx}}$ and set the last $r-1$ entries of both $\bfx$ and  ${\widehat {\bfx}}$ to $0$ it follows that $\bfB \bfx = {\bf \Pi} \bfB \widehat{\bfx}$. 
\end{proof}

Now we extend the converse result to the case of odd number of measurements $N$ where $N$ is less than $2K$. We first prove the following lemma.

\begin{lemma}\label{odd_lemma}
Let $\bfB$ be an $N\times K$ matrix with $N=2K-1$ being an odd number less than $2K$. Let ${\bf \Pi}$ be an $N\times N$ permutation matrix with only one cycle, e.g., ${\bf \Pi}_{i,j}=1$ only for those $(i,j) \in \{(1,2), (2,3), \dots, (N-1,N), (N,1)\}$. Let $\bfG=[\bfB, {\bf \Pi} \bfB]$ and let $\widetilde{\bfG}$ be the $(2K-1)\times (2K-1)$ submatrix of $\bfG$ obtained by dropping the last column of $\bfG$. Then, there is an assignment to the elements of matrix $\bfB$ such that $\det(\widetilde{\bfG})\neq 0$. 
\end{lemma}

\begin{proof}
To simplify the proof, note that we can equivalently represent the submatrix $\widetilde{\bfG}$ as follows
\begin{align}
\widetilde{\bfG}=\left[ \begin{array}{cc} \bfv^1 & \overline{\bfv}^2 \\ \bfv^2 & \overline{\bfv}^3\\ \vdots & \vdots \\ \bfv^{2K-2} & \overline{\bfv}^{2K-1}\\  \bfv^{2K-1} & \overline{\bfv}^1 \end{array} \right ],
\end{align}
where $\bfv^i$, $i \in [2K-1]$, are $K$-dimensional vectors corresponding to the rows of $\bfB$ and where $\overline{\bfv}^i$ is the $(K-1)$-dimensional vector obtained after dropping the last component of $\bfv^i$. We can simply check that the following assignment gives a nonzero determinant for $\widetilde{\bfG}$. We take $\bfv^1=\bfd^K$, and for $i=2,3,\dots, 2K-1$, we set
\begin{align}
\bfv^i=\left \{ \begin{array}{cc} {\bf 0} & \text{$i$ even,}\\ \bfd^{\frac{i-1}{2}} & \text{$i$ odd.}\end{array} \right.
\end{align}
where $\{\bfd^k\}_{k=1}^K$ denotes a standard basis for $\bR^K$ with $\bfd^k \in \bR^K$ having an entry of $1$ as its $k$-th component and $0$ elsewhere. We can check that for this assignment $\widetilde{\bfG}$ has exactly one $1$ in each row and column and all other elements equal to $0$. Thus, it is a permutation matrix and as a result it is invertible.
\end{proof}

\begin{lemma}\label{odd_lemma_2}
Assume that conditions of Lemma \ref{odd_lemma} hold. Suppose that the components of $\bfB$ are sampled i.i.d. from a continuous distribution, then $\det(\widetilde{\bfG})\neq 0$ w.p.  $1$.
\end{lemma}

\begin{proof}
The proof simply follows from the proof of Lemma \ref{odd_lemma}. Note that $\det(\widetilde{\bfG})$ is a polynomial of the components of $\bfB$. From Lemma \ref{odd_lemma}, there is an assignment to the matrix $\bfB$ with a nonzero $\det(\widetilde{\bfG})$. As the components of $\bfB$ are sampled i.i.d. from a continuous distribution, from Lemma \ref{lem:zeroprob}, it results that $\det(\widetilde{\bfG})\neq 0$ w.p. $1$.
\end{proof}

Now we prove the converse result for odd number of measurements $N<2K$. 

\begin{proposition}\label{prop:converse_odd}
Let $\bfB$ be an $N\times K$ matrix with $K>1$ and $N<2K$ being an odd number whose components are drawn i.i.d. at random from a continuous distribution. Then, w.p. $1$, 
there exists a permutation matrix ${\bf \Pi}$ and vectors $\bfx, \widehat{\bfx}$ such that $\bfx \neq \widehat{\bfx}$ and $\bfB\, \bfx = {\bf \Pi} \bfB \widehat \bfx$. 
\end{proposition}

\begin{proof}

Similar to the proof of even number of measurements as in Proposition \ref{prop:converse_even_2}, we will show that we can always find a permutation ${\bf \Pi}$, and two vectors $\bfx$ and $\widehat{\bfx}$ for which $\bfB\bfx={\bf \Pi} \bfB \widehat{\bfx}$ but $\bfx \neq \widehat \bfx$. 
We first consider the simple case where $N=2K-1$. 
We fix the permutation matrix ${\bf \Pi}$ as the one described in the statement of Lemma \ref{odd_lemma}, which we recall has only one cycle. 
Let $\bfG=[\bfB, {\bf \Pi} \bfB]$. 
Let $\bfg$ be the last column of $\bfG$ and let $\widetilde{\bfG}$ be the $(2K-1)\times (2K-1)$ matrix obtained from $\bfG$ after removing $\bfg$, thus, $\bfG=[\widetilde{\bfG}, \bfg]$. 

Let $\bft = \widetilde{\bfG}^{-1} \bfg$.
We will first show that the $K$-th component of $\bft$ satisfies $\bft_K \neq 1$ w.p.  $1$.
Let $\Delta := \bft_K-1$. 
We will first verify the following claim.

\medskip

\noindent \emph{Claim:} The vector $\bft$ is well-defined and $\Delta \neq 0$ w.p. 1.

\noindent We prove the claim as follows.
Clearly, from the definition of $\bft$ we have $\Delta= (\bfu^K) ^\transp \widetilde{\bfG}^{-1} \bfg - 1$, where $\{\bfu^\ell\}_{\ell=1}^{2K-1}$ denotes the standard basis for $\bR^{2K-1}$ with $\bfu^\ell$ being a unit vector with only one $1$ as its $\ell$-th component and zero elsewhere. 
It is not difficult to see that $\Delta$ is a rational function of the components of the matrix $\bfB$, i.e., $\Delta=\frac{p(\bfB)}{q(\bfB)}$ where $p$ and $q$ are polynomials of the components of $\bfB$. 
Furthermore, without loss of generality, $p$ and $q$ are relatively prime. 
Moreover, as the inverse of the matrix is $\widetilde{\bfG}^{-1} = \frac{\mbox{Adjoint of }\widetilde{\bfG}}{\det(\widetilde{\bfG})}$, it follows from the definition of $\Delta$ that $\det(\widetilde{\bfG})$ expressed as a polynomial of $\bfB$ must be divisible by $q(\bfB)$. 
As the components of $\bfB$ are sampled i.i.d. from a continuous distribution, from Lemma \ref{odd_lemma_2}, it results that $\det(\widetilde{\bfG})\neq 0$ w.p. $1$, which implies that $q(\bfB)\neq 0$ w.p. $1$. 
Hence, $\bft$ and $\Delta$ are almost surely well-defined. 
Now it remains to prove that $\Delta \neq 0$ w.p. $1$. To prove this, from Lemma \ref{lem:zeroprob}, we simply need to find an assignment to $\bfB$ such that {$p(\bfB)\neq 0$ which implies that }$\Delta \neq 0$. 
Note that for our choice of $\bf \Pi$ the matrix $\widetilde{\bfG}$ is given by
\begin{align}
\widetilde{\bfG}=\left[ \begin{array}{cc} \bfv^1 & \overline{\bfv}^2 \\ \bfv^2 & \overline{\bfv}^3\\ \vdots & \vdots \\ \bfv^{2K-2} & \overline{\bfv}^{2K-1}\\  \bfv^{2K-1} & \overline{\bfv}^1 \end{array} \right ],
\end{align}
where $\bfv^i$, {$i\in[2K-1]$}, denote the rows of the matrix $\bfB$, and where $\overline{\bfv}^i$ is a $(K-1)$-dimensional vector obtained after dropping the last component of $\bfv^i$. 
Similar to the proof of Lemma \ref{odd_lemma}, we consider the following assignment. We take $\bfv^1=\bfd^K$, and for $i=2,3,\dots, 2K-1$, we set
\begin{align}
\bfv^i=\left \{ \begin{array}{cc} {\bf 0} & \text{$i$ even,}\\ \bfd^{\frac{i-1}{2}} & \text{$i$ odd.}\end{array} \right.
\end{align}
where $\{\bfd^k\}_{k=1}^K$ is the standard basis for $\bR^K$. From the proof of Lemma \ref{odd_lemma}, it immediately results that for this assignment $\widetilde{\bfG}$ is a permutation matrix. Consequently, we have $|\det(\widetilde{\bfG})|=1\neq 0$ and $\widetilde{\bfG}^{-1}=\widetilde{\bfG}^\transp$. Moreover, for the same assignment, we have $\bfg=[{\bf 0}_{2K-2}\, ; 1]=\bfu^{2K-1}$, and it is not  difficult to also check that the row $(K-1)$ in $\widetilde{\bfG}^\transp$ is equal to $(\bfu^{2K-1})^\transp$. As the rows $\widetilde{\bfG}^\transp$ are orthonormal, it results that $\widetilde{\bfG}^{-1} \bfg= \widetilde{\bfG}^\transp \bfu^{2K-1}= \bfu^{K-1}$. Hence, we obtain $(\bfu^K) ^\transp \widetilde{\bfG}^{-1} \bfg= (\bfu^K)^\transp \bfu^{K-1}=0$. This implies that $\Delta= (\bfu^K) ^\transp \widetilde{\bfG}^{-1} \bfg - 1 = -1$.
Thus there is an assignment to the entries of $\bfB$ such that $\Delta \neq 0$ and hence $p(\bfB) \neq0$.
By Lemma \ref{lem:zeroprob} this implies that $p(\bfB) \neq0$ w.p. $1$. 
Since in addition, we have already established that $q(\bfB) \neq0$ w.p. $1$ it follows that $\Delta \neq0$ w.p. $1$. 
This completes the proof of the claim.

Let us define vectors $\bfx$ and $\widehat \bfx$ as follows.
Let 
\begin{eqnarray}
\bfx_i  = \bft_i, \mbox{ for } { i \in [K]} \qquad \mbox{ and }\qquad  \widehat \bfx_i= \left \{ 
\begin{array}{cc}
-\bft_{K+i}, & \mbox{for } {i \in [K-1]}\\
1, & \mbox{for }i = K
\end{array}
\right.
\end{eqnarray}
With this assignment, it is clear that, w.p. 1,
\[
\bfB \bfx - \bf \Pi \bfB \widehat \bfx = \widetilde{\bfG} \bft - g = 0.
\]
Furthermore, by the result of the claim $\bfx_K - \widehat \bfx_K=\bft_K-1 =\Delta \neq 0$, and thus $\bfx \neq \widehat \bfx$  w.p. 1. This completes the proof of the proposition for $N = 2K-1$.

If $N=2K-2r-1$ is an odd number less than $2K-1$, using a similar idea as in Proposition \ref{prop:converse_even_2}, we can build a signal $\bfx$ by setting the last $r$ components of $\bfx$ equal to zero and applying our proof for $N=2K-2r-1$ and $K'=K-r$ to show that the recovery algorithm fails to find the correct solution w.p.  $1$.
This completes the proof.

\end{proof}


Combining the results of Propositions \ref{prop:converse_even_2} and \ref{prop:converse_odd} we complete the proof of the converse result {given in} Theorem \ref{thm:converse}. 

\section{Local stability under additive noise}\label{sec:stability}
In practice, the observation $\bfy$ is typically corrupted by noise.
Consider a noisy version of the linear system of \eqref{eqn:permobs2}.
Let
\[
\bfy=\bfB \bfx^0 + \bfw
\]
be a noisy measurement of the signal $\bfx^0$ with an additive noise $\bfw$. We assume that $\bfB=\bfS^0 \bfA$ for some selection matrix $\bfS^0 \in \clS$ where $\clS$ 
denotes the set of all selection matrices, i.e., the set of all matrices comprising $N$ distinct rows of the $M \times M$ identity matrix arranged in any arbitrary order.
We define the \textit{signal-to-noise ratio} (SNR) for the given measurement $\bfy$ by $\snr=\frac{\|\bfB \bfx^0\|^2}{\|\bfw\|^2}$. 
In the noisy case, a natural reconstruction algorithm is the following robust version of the original algorithm 
\begin{align}\label{rob_rec_alg}
\widehat{\bfx}^0=\argmin_{\bfx \in \bR^K} \min_{\bfS \in \clS} \|\bfy - \bfS \bfA \bfx\|.
\end{align}
The reconstruction error is given by $\|\bfx^0 - \widehat{\bfx}^0\|$. 
We call the recovery algorithm \eqref{rob_rec_alg} locally stable if for an arbitrary signal $\bfx^0 \in \bR^K$ and measurement noise $\bfw$ with an $\snr=\frac{\|\bfB \bfx^0\|^2}{\|\bfw\|^2}$, we have $\lim_{\snr \to \infty} \|\bfx^0 -\widehat{\bfx}^0\|=0$.
We will now argue that for the random design of $M\times K$ matrix $\bfA$ introduced earlier, the recovery algorithm \eqref{rob_rec_alg} is locally stable.

From Theorem~\ref{thm:main_theorem} we know that w.p.  $1$ any choice of $\bfA$ satisfies the two properties below. \

\begin{assumption}
\item \label{ass:unique}If the relation $\bfS' \bfA \bfx'=\bfS'' \bfA \bfx''$ is satisfied for some $\bfS',\bfS'' \in \clS$, then $\bfx'=\bfx''$. 

\item \label{ass:rank}Columns of $\bfA$ are linearly independent.
\end{assumption}
Suppose that $\bfA$ is such a matrix. We define a distance over $\clS$ as follows. For a given $\bfS', \bfS'' \in \clS$, 
We define the minimum principal angle between the subspaces $\range(\bfS'\bfA)$ and $\range(\bfS''\bfA)$ as follows 
\begin{align}
\theta^*({\bfS' \bfA,\bfS''\bfA})=\cos^{-1}\left(\max \{\langle \bfu', \bfu'' \rangle: \bfu' \in \range(\bfS'\bfA),  \bfu'' \in \range(\bfS''\bfA), \|\bfu'\| = \|\bfu''\| = 1\} \right).
\end{align}
Using $\theta^*$, we define the distance between $\range(\bfS' \bfA)$ and $\range(\bfS'' \bfA)$ as $d(\bfS' \bfA,\bfS''\bfA)=\sin(\theta^*({\bfS' \bfA,\bfS''\bfA}))$. 
By the symmetry of the definition, $d$ is also symmetric, i.e., $d(\bfS'\bfA,\bfS''\bfA)=d(\bfS''\bfA,\bfS'\bfA)$, and $d(\bfS'\bfA,\bfS''\bfA)=0$ if and only if the subspaces $\range(\bfS'\bfA)$ and $\range(\bfS''\bfA)$ have a nontrivial intersection, i.e., if their intersection is a subspace of dimension $1$ or larger. 
However $d$ is not strictly a distance measure as it does not satisfy the triangle inequality in general.

\begin{proposition}
Let $\bfx^0 \in \bR^K$ and $\bfB = \bfS \bfA$ where $\bfS$ is a selection matrix as before and $\bfA$ satisfies assumptions \Ass{unique} and \Ass{rank}.
Let $\widehat{\bfx}^0$ be the output of the algorithm \eqref{rob_rec_alg} for the measurement $\bfy=\bfB \bfx^0 +\bfw$, where $\bfw$ is an arbitrary noise vector with $\|\bfw\| = \|\bfB \bfx^0\| \snr^{-\half}$, {where $\snr$ is the \textit{signal-to-noise ratio}}. Then, $\lim_{\snr \to \infty} \|\bfx^0 -\widehat{\bfx}^0\|=0$.

\end{proposition}
\begin{proof}
Let $\bfy^0=\bfB \bfx^0$. For simplicity, for an arbitrary $\bfS$, we use slightly abusive notations of $\theta^*(\bfy^0, \bfS \bfA)$ and $d(\bfy^0, \bfS \bfA)$ to denote the angle and distance between the subspaces $\range(\bfy^0)$ and $\range(\bfS \bfA)$ respectively. Consider the partition $\clS=\clS_1\cup\clS_2$, where $\clS_1=\{\bfS \in \clS: d(\bfy^0, \bfS \bfA)=0\}$ and $\clS_2=\{\bfS \in \clS: d(\bfy^0, \bfS \bfA)>0\}$, and let $d_{\min}(\clS_2)=\min_{\bfS \in \clS_2} d(\bfy^0, \bfS \bfA)$. Note that $d_{\min}(\clS_2)=\sin(\theta_{\min})$, where $\theta_{\min}$ is the minimum angle between the subspace $\range(\bfy^0)$ and any subspace $\range(\bfS\bfA)$ such that $\bfS$ is in $\clS_2$. 

%

Let $\widehat{\bfS}$ denote the optimal choice of the selection matrix in the optimization of \eqref{rob_rec_alg} under the optimal choice of $\bfx =\widehat{\bfx}^0$.
Then we have
\begin{align*}
\widehat{\bfS}&=\argmin_{\bfS \in \clS} \min_{\bfx \in \bR^K} \|\bfy - \bfS \bfA \bfx\|\\
&=\argmin_{\bfS \in \clS} \min_{\bfx \in \bR^K} \|\bfy\|^2 - 2\inp{\bfy}{\bfS\bfA \bfx} + \|\bfS \bfA \bfx\|^2\\
&=\argmin_{\bfS \in \clS} \min_{\bfu \in \range(\bfS\bfA)} \{ - 2\inp{\bfy}{\bfu} + \|\bfu\|^2\}\\
&=\argmin_{\bfS \in \clS} \min _{\lambda \in \bR_+} \Big \{ \min_{\bfu \in \range(\bfS\bfA): \|\bfu\|=\lambda}\big \{- 2\inp{\bfy}{\bfu} + \|\bfu\|^2\big \}\Big \}\\
&=\argmin_{\bfS \in \clS}  \min _{\lambda \in \bR_+} \lambda^2 - 2\lambda \|\bfy\| \cos(\theta^*(\bfy, \bfS \bfA))\\
&=\argmin_{\bfS \in \clS} - \|\bfy\|^2 \cos^2(\theta^*(\bfy,\bfS \bfA))\\
&=\argmin_{\bfS \in \clS} d(\bfy,\bfS \bfA).
\end{align*}
Hence, it results that the optimal selection matrix $\widehat{\bfS}$ minimizes $d(\bfy, \bfS \bfA)$. 
With a geometric argument, it is not difficult to show that as long as $\frac{\|\bfy^0\|}{\|\bfw\|}> \frac{1}{\sin(\theta_{\min}/2)}$, the selection matrix $\widehat{\bfS}$ can not belong to $\clS_2$, since at least for the selection matrix $\bfS^0 \in \clS_1$, we have $d(\bfy, \bfS^0 \bfA)<\sin(\theta_{\min}/2) <d(\bfy, \bfS \bfA)$ for every $\bfS \in \clS_2$. This has been illustrated in Fig.~\ref{sub_dist} for the simple case of $1$-dimensional subpaces.

\begin{figure}[h]
\centering
\begin{tikzpicture}[scale=0.8]
\draw[->] (-4,0) -> (4,0);
\draw[->] (0,-4) -> (0,4);
\draw[thick, red] (-2,-4) -- (2,4) node[above, red] {$\range(\bfy_0)$};
\draw[thick, blue] (-4,-2) -- (4,2) node[right, blue]{$\range(\bfS \bfA)$};
\draw[thick, ->] (1,0.5) arc (30:65:1.2);
\node at (1.7,0.4) {$\theta_{\min}$};
\draw[dashed] (0,0) -- (4,4);

\draw[fill=black] (1.5,3) circle (.1) node[above left] {$\bfy^0$};
\draw[black] (1.5,3) circle (0.8);
\draw[thick, ->] (1.5,3) -- (2.3,3);
\node at (2.5, 3.5) {$\bfw$};
\end{tikzpicture}
\caption{A geometric view of distance from the subspace $\range(\bfy_0)$ to another subspace $\range(\bfS \bfA)$ for some $\bfS \in \clS_2$.}
\label{sub_dist}
\end{figure}
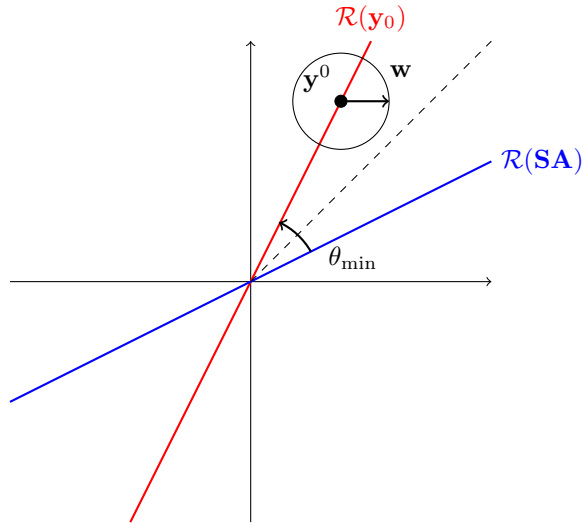


This implies that for sufficiently large $\snr$, we need to consider only those $\bfS \in \clS_1$. 
For sufficiently high $\snr$, let $\widehat{\bfS} \in \clS_1$ be the optimal selection matrix obtained in the optimization problem in \eqref{rob_rec_alg} under the optimal choice of $\bfx =\widehat{\bfx}^0$. Then it results that 
\begin{align}
\widehat{\bfx}^0=\argmin_{\bfx \in \bR^K} \| \bfy^0 + \bfw - \widehat{\bfS} \bfA \bfx\|. 
\end{align}
Note that since $\widehat{\bfS} \in \clS_1$, we have $d(\bfy^0, \widehat{\bfS} \bfA)=0$. This implies that $\bfy^0$ is in the column span of $\widehat{\bfS} \bfA$. In other words, there is some $\bfx'\in \bR^K$ such that 
$\bfy^0=\bfB \bfx^0=\bfS \bfA \bfx^0=\widehat{\bfS} \bfA \bfx'$. From assumption \Ass{unique}, it immediately follows that $\bfx'=\bfx^0$ and thus, $\bfy^0=\widehat{\bfS}\bfA \bfx^0$. Hence we have 
\begin{align}
\widehat{\bfx}^0&=\argmin_{\bfx \in \bR^K} \| \bfy^0 + \bfw - \widehat{\bfS} \bfA \bfx\|=\argmin_{\bfx \in \bR^K} \| \widehat{\bfS} \bfA \bfx^0 + \bfw - \widehat{\bfS} \bfA \bfx\|=(\widehat{\bfS} \bfA)^\dagger(\bfy^0 + \bfw)
\end{align}
where we use the fact that by \Ass{rank} the matrix $\bfA$ and thus $\widehat{\bfS} \bfA$ has full column rank.
Thus
\begin{align}
\widehat{\bfx}^0= (\widehat{\bfS} \bfA)^\dagger  (\widehat{\bfS} \bfA \bfx^0 + \bfw)
= \bfx^0 + (\widehat{\bfS} \bfA)^\dagger \bfw,
\end{align}
which implies that 
$
\|\bfx^0-\widehat{\bfx}^0\| = \|(\widehat{\bfS} \bfA)^\dagger \bfw\|.
$
Since $\widehat{\bfS} \in \clS_1$ and there are only finitely many possible subset selections in $\clS_1$, this implies that as $\snr$ tends to infinity, $\|\bfx^0-\widehat{\bfx}^0\|$ tends to zero, which confirms the local stability of the recovery algorithm.
\end{proof}
Thus we have proved the local stability of the recovery algorithm to noise.
In particular, for any $\bfx$, the estimate $\widehat \bfx$ converges to the true $\bfx$ as SNR tends to infinity. 
It remains to be seen whether global stability is satisfied.
It might be possible to show, for instance, that as SNR tends to infinity, the worst case error among all possible $\bfx$ tends to zero.
However, we leave this issue for further investigations.

\section{Conclusions}\label{sec:conc}
We studied the problem of unlabeled sensing defined as the problem of inverting a linear system with unlabeled observations.
We showed that an oversampled linear system with random coefficients can be inverted from unlabeled measurements provided the oversampling ratio is $2$ or higher.
Moreover, for oversampling ratios greater than $2$ any $2K$ measurements are sufficient to recover an unknown vector of $K$ elements.
We also obtained the converse result that $2K$ is the minimum number of measurements needed for the result to hold and demonstrated local stability of the recovery algorithm to noise.
In essence, the main result presented here is a geometric property of random matrices. 
It is of interest to see whether the geometry of the result can be further understood.
These results also raise a number of interesting follow-up questions, including, whether a faster algorithm exists for recovering $\bfx$ and whether it is possible to learn $\bfA$ using a training phase with unlabeled observations from known $\bfx$'s.
It is also of interest to identify examples of deterministic designs of $\bfA$ that admit recovery from unlabeled measurements.
Another interesting question is whether the recovery algorithm satisfies global stability in the presence of noise, as discussed in the concluding parts of Section \ref{sec:stability}.
These are topics of current research.


\bibliographystyle{IEEEbib}
\bibliography{unlabeledrefs}

\begin{thebibliography}{10}

\bibitem{unnhagvet15}
J.~Unnikrishnan, S.~Haghighatshoar, and M.~Vetterli,
\newblock ``Unlabeled sensing: Solving a linear system with unordered
  measurements,'' presented at the 2015 53rd Annual Allerton Conference on
  Communication, Control, and Computing, Oct 2015.

\bibitem{pravet08}
Paolo Prandoni and Martin Vetterli,
\newblock {\em Signal {P}rocessing for {C}ommunications},
\newblock CRC Press, 2008.

\bibitem{robtre80}
Enders~A Robinson and Sven Treitel,
\newblock {\em Geophysical {S}ignal {A}nalysis}, vol. 263,
\newblock Prentice-Hall New Jersey, 1980.

\bibitem{aka12}
Metin Akay,
\newblock {\em Biomedical signal processing},
\newblock Academic Press, 2012.

\bibitem{vetkovgoy14}
M.~Vetterli, J.~Kova{\v{c}}evi{\'c}, and V.K. Goyal,
\newblock {\em Foundations of Signal Processing},
\newblock Cambridge University Press, 2014.

\bibitem{don06}
David~L. Donoho,
\newblock ``Compressed sensing,''
\newblock {\em IEEE Trans. Inform. Theory}, vol. 52, pp. 1289--1306, 2006.

\bibitem{cantao06}
E.J. Candes and T.~Tao,
\newblock ``Near-optimal signal recovery from random projections: Universal
  encoding strategies?,''
\newblock {\em Information Theory, IEEE Transactions on}, vol. 52, no. 12, pp.
  5406--5425, Dec 2006.

\bibitem{thrleo08}
Sebastian Thrun and John~J. Leonard,
\newblock ``Simultaneous localization and mapping,''
\newblock in {\em Springer Handbook of Robotics}, Bruno Siciliano and Oussama
  Khatib, Eds., pp. 871--889. Springer Berlin Heidelberg, 2008.

\bibitem{bla86}
Samuel~S. Blackman,
\newblock {\em Multiple-target tracking with radar applications},
\newblock Artech House radar library. Norwood, Mass. Artech House, 1986.

\bibitem{unnvet13}
J.~Unnikrishnan and M.~Vetterli,
\newblock ``Sampling and reconstruction of spatial fields using mobile
  sensors,''
\newblock {\em Signal Processing, IEEE Transactions on}, vol. 61, no. 9, pp.
  2328--2340, May 2013.

\bibitem{unnvet13b}
J.~Unnikrishnan and M.~Vetterli,
\newblock ``Sampling high-dimensional bandlimited fields on low-dimensional
  manifolds,''
\newblock {\em Information Theory, IEEE Transactions on}, vol. 59, no. 4, pp.
  2103--2127, April 2013.

\bibitem{bal62}
A.~Balakrishnan,
\newblock ``On the problem of time jitter in sampling,''
\newblock {\em Information Theory, IRE Transactions on}, vol. 8, no. 3, pp.
  226--236, April 1962.

\bibitem{bro07}
J.~Browning,
\newblock ``Approximating signals from nonuniform continuous time samples at
  unknown locations,''
\newblock {\em Signal Processing, IEEE Transactions on}, vol. 55, no. 4, pp.
  1549--1554, April 2007.

\bibitem{marvet00}
P.~Marziliano and M.~Vetterli,
\newblock ``Reconstruction of irregularly sampled discrete-time bandlimited
  signals with unknown sampling locations,''
\newblock {\em Signal Processing, IEEE Transactions on}, vol. 48, no. 12, pp.
  3462--3471, Dec 2000.

\bibitem{emidaugri14}
V.~Emiya, A.~Bonnefoy, L.~Daudet, and R.~Gribonval,
\newblock ``Compressed sensing with unknown sensor permutation,''
\newblock in {\em Acoustics, Speech and Signal Processing (ICASSP), 2014 IEEE
  International Conference on}, May 2014, pp. 1040--1044.

\bibitem{narshm08}
Arvind Narayanan and Vitaly Shmatikov,
\newblock ``Robust de-anonymization of large sparse datasets.,''
\newblock in {\em IEEE Symposium on Security and Privacy}. 2008, pp. 111--125,
  IEEE Computer Society.

\bibitem{naiunnthivet15}
F.~M. Naini, J.~Unnikrishnan, P.~Thiran, and M.~Vetterli,
\newblock ``Where you are is who you are: User identification by matching
  statistics,''
\newblock {\em Information Forensics and Security, IEEE Transactions on}, vol.
  PP, no. 99, pp. 1--1, 2015.

\bibitem{fogjenbacdas13}
Fajwel Fogel, Rodolphe Jenatton, Francis Bach, and Alexandre d'Aspremont,
\newblock ``Convex relaxations for permutation problems,''
\newblock in {\em Advances in Neural Information Processing Systems 26}, C.J.C.
  Burges, L.~Bottou, M.~Welling, Z.~Ghahramani, and K.Q. Weinberger, Eds., pp.
  1016--1024. Curran Associates, Inc., 2013.

\bibitem{schzuc99}
L.J. Schulman and D.~Zuckerman,
\newblock ``Asymptotically good codes correcting insertions, deletions, and
  transpositions,''
\newblock {\em Information Theory, IEEE Transactions on}, vol. 45, no. 7, pp.
  2552--2557, Nov 1999.

\bibitem{sha48}
Claude~E. Shannon,
\newblock ``{A} {M}athematical {T}heory of {C}ommunication,''
\newblock {\em The Bell System Technical Journal}, vol. 27, no. 3, pp.
  379--423, 1948.

\end{thebibliography}

\appendices
\section{Proof of Lemma \ref{lem:zeroprob}}\label{app:lem:zeroprob}
In this section, we prove Lemma \ref{lem:zeroprob}. First we need to define some notation. We denote by $\clV_n=\{x_1,x_2, \dots, x_n\}$ the set of $n$ variables. A monomial in $\clV_n$ is a product of the form $\prod_{k=1}^n x_k ^{i_k}$ where all the exponents $i_1, i_2, \dots, i_n$ are nonnegative integers.  For simplicity, we define $I=(i_1,i_2, \dots, i_n)$ and set $x^I=\prod_{k=1}^n x_k ^{i_k}$. The degree of this multinomial is defined by $|I|=\sum_{k=1}^n i_k$. A multinomial in $\clV_n$ with coefficients in $\bR$ is a finite linear combination of monomials. We write a multinomial $p$ in the form 
\begin{align}
p(x_1, x_2, \dots, x_n)=\sum_{I} a_I x^I, \ a_I \in \bR.
\end{align}
We define the degree of $p$ as the maximum degree of its constituents monomials and denote it by $\mathsf{deg}(p)$. We denote by $\partial_k p=\frac{\partial p}{\partial x_k}$, $k=1,2,\dots,n$, the multinomial obtained by taking the partial derivative of $p$ with respect to $x_k$.

\begin{proposition}\label{prop:zerovolume}
Let $p$ be a nonzero multinomial over the variables $\clV_n$. Let $\clZ$ denote the zero-set of $p$, namely, 
\begin{align}
\clZ=\{(x_1,x_2,\dots,x_n) \in \bR^n: p(x_1,x_2,\dots,x_n)=0\}.
\end{align}
Then $\clZ$ is Lebesque-measurable with $\lambda(\clZ)=0$, where $\lambda$ denotes the Lebesgue measure in $\bR^n$. 
\end{proposition}

\begin{proof}
First notice that $p$ is a continuous function from $\bR^n$ to $\bR$. Moreover, the zero-set $\clZ$ can be simply written as $p^{-1}(\{0\})$, where $p^{-1}$ denotes the inverse image of $p$. As $\{0\}$ is a closed set in $\bR$, from the continuity of $p$, it results that $\clZ$ is a closed set in $\bR^n$. Hence, it is Lebesgue-measurable.
Now, we use  induction on the degree of $p$ to show that $\lambda(\clZ)=0$. 

If $\mathsf{deg}(p)=0$, then $p=a_0$ is a constant term with $a_0\neq 0$. In this case the zero-set of $p$ is empty and the result holds. If $\mathsf{deg}(p)=1$ then $p=a_0 + \sum_{k=1}^n a_k x_k$ and the zero-set of $p$ is a hyperplane in dimension $n$. 
As the Lebesgue measure is rotation and translation invariant, the measure of the zero-set of $p$ is equal to the measure of the set $\{(x_1,x_2,\dots,x_n)\in \bR^n: x_1=0\}$, which we know has zero Lebsegue measure.

Now assume that $\mathsf{deg}(p)\geq 2$. Let us define $n$ multinomials $p_k=\partial _k p$ for $k=1,2, \dots,n$. As $\mathsf{deg}(p)\geq 2$, there should be at least one nonzero $p_k$, say $p_1$, where $\mathsf{deg}(p_1)=\mathsf{deg}(p)-1\geq 1$. Let $\clZ_1$ denote the zero-set of $p_1$. From the induction hypothesis, it results that $\lambda(\clZ_1)=0$ and as a result $\lambda(\clZ \cap \clZ_1)=0$. Also let $\clW=\clZ\cap \clZ_1^c$ be the set of all points in the zero-set of $p$ that are not included in the zero-set of $p_1$. If $\clW$ is empty then $\lambda(\clW)=0$. Otherwise, let $x^*:=(x^*_1,x^*_2, \dots, x^*_n)\in \clW$ be an arbitrary point. Since $p(x^*)=0$ but $\frac{\partial}{\partial x_1}p(x^*)=p_1(x^*) \neq 0$, from the implicit function theorem, there is an open neighbourhood $\clO \subset \bR^{n-1}$ containing $(x^*_2, x^*_3, \dots, x^*_n)$ and an open interval $\clI \subset \bR$ containing $x^*_1$, and a differentiable function $g: \clO \to \clI$ such that $x^*_1=g(x^*_2, x^*_3, \dots, x^*_n)$, and $p\big (g(x_2,x_3, \dots, x_n), x_2, x_3, \dots, x_n\big )=0$, for all $(x_2,x_3, \dots, x_n)\in \clO$. 
Let $\clB$ be a rectangular open ball, i.e., $\clB=\{x \in \bR^n: \|x-c\|_\infty < r\}$ for some $c \in \bR^n$ and some $r>0$, that contains $x^*$ and is contained in the open set $\clI \times \clO$. Without loss of generality, assume that the center of this ball $c$ has rational coordinates and its radius $r$ is also rational. We have
\begin{align}
\lambda(\clW\cap \clB)&= \int _{\clW\cap \clB} \ind_{x_1=g(x_2, \dots, x_n)} dx_1 dx_2 \dots dx_n \leq \int _{\clB} \ind_{x_1=g(x_2, \dots, x_n)} dx_1 dx_2 \dots dx_n\\
&= \int_{\widetilde{\clB}} \Big \{\int _{c_1 -r}^{c_1 +r} \ind_{x_1=g(x_2, \dots, x_n)} dx_1 \Big\} dx_2 \dots dx_n=\int_{\widetilde{\clB}} 0\ dx_2 \dots dx_n=0,
\end{align}
where $\widetilde{\clB}=\{ (x_2, \dots, x_n) \in \bR^{n-1}: |x_i -c_i| <r, \ i=2,3, \dots, n\}$ denotes the projection of the $n$-dimensional ball $\clB$ on its last $n-1$ components. 
Let us denote the set of all such rectangular balls corresponding to the points of $\clW$ with $\frkB$. Note that $\frkB$ is a countable set whose elements can be enumerated by $\frkB=\{\clB_1,\clB_2, \dots\}$, where for every $i=1,2,\dots$, we have $\lambda(\clW\cap \clB_i)=0$. Moreover, $\cup_{i=1}^\infty \clB_i$ covers the set $\clW$ since every point $x^*\in\clW$ is contained in at least one of these balls, thus, we have
\begin{align}
\lambda(\clW) =\lambda \Big (\clW \cap \, \cup_{i=1}^\infty \clB_i \Big ) = \lambda\Big (\cup_{i=1}^\infty (\clW \cap \clB_i) \Big )\leq \sum_{i=1}^\infty \lambda(\clW\cap \clB_i)=0. 
\end{align}
Hence, we obtain $\lambda(\clW)=\lambda(\clZ \cap \clZ_1^c)=0$, which together with $\lambda(\clZ \cap \clZ_1)=0$ implies that $\lambda(\clZ)=0$. This completes the induction step and proves the result.
\end{proof}

In this paper, we are interested in a probabilistic version of Proposition \ref{prop:zerovolume} stated in the following proposition.

\begin{proposition}\label{prop:zeroprob}
Let $p$ be a nonzero multinomial over the variables $\clV_n$ with the zero-set $\clZ$. Let $\bP$ be an arbitrary continuous probability distribution over $\bR^n$. Then $\bP(\clZ)=0$.
\end{proposition}

\begin{proof}
From Proposition \ref{prop:zerovolume}, we obtain that $\lambda(\clZ)=0$. As $\bP$ is a continuous probability distribution, by definition, it is dominated by the Lebesgue measure $\lambda$, thus, $\lambda(\clZ)=0$ implies that $\bP(\clZ)=0$.
\end{proof}

In the special case where $\bP$ is the product of one dimensional continuous distributions, we obtain the proof of Lemma \ref{lem:zeroprob}. \vspace{1mm} 

{\bf Proof of Lemma \ref{lem:zeroprob}:} Let $p$ be the nonzero multinomial, let $X_1, X_2, \dots, X_n$ be the sequence of $n$ independent (not necessarily identically distributed) real-valued random variables, and denote by $\bP_i$ the probability distribution of $X_i$. Since $X_i$ are independent, it results that $(X_1, X_2, \dots, X_n)$ has the product distribution $\bP=\bP_1\times \bP_2 \times \dots \times \bP_n$. As each $\bP_i$ is continuous with respect to the $1$-dimensional Lebesgue measure, $\bP$ will be continuous with respect to the $n$-dimensional Lebesgue measure $\lambda$, which implies that $(X_1, X_2, \dots, X_n)$ has a continuous distribution. Thus, using the Proposition \ref{prop:zeroprob}, we obtain the result.  \hfill $\square$

\end{document}